\documentclass{article}

\title{Pinned Distance Sets Using Effective Dimension}
\author{D. M. Stull\\Department of Computer Science, Northwestern University\\
	Evanston, IL 60208, USA\\
	\texttt{donald.stull@northwestern.edu}}

\usepackage{amsmath}
\usepackage{amssymb}
\usepackage{mdframed}
\usepackage{amsthm}
\usepackage{mathtools}
\usepackage{enumitem}
\usepackage{placeins}
\usepackage{url}
\usepackage{hyperref}

\newtheorem{thm}{Theorem}
\newtheorem{obs}[thm]{Observation}
\newtheorem{lem}[thm]{Lemma}
\newtheorem{prop}[thm]{Proposition}
\newtheorem{cor}[thm]{Corollary}

\newtheorem*{T1}{Theorem~\ref{thm:maintheorem}}
\newtheorem*{T2}{Theorem~\ref{thm:projectionMainTheorem}}
\newtheorem*{T3}{Theorem~\ref{thm:DistanceLowerBoundComplexity}}
\newtheorem*{T4}{Theorem~\ref{thm:mainThmEffDim}}
\newtheorem*{L1}{Lemma~\ref{lem:reductionToEff}}

\theoremstyle{remark}
\newtheorem*{remark}{Remark}

\DeclareMathOperator{\dimH}{dim_H}

\newcommand{\R}{\mathbb{R}}

\newcommand{\N}{\mathbb{N}}
\newcommand{\Q}{\mathbb{Q}}

\newcommand{\ve}{\varepsilon}
\newcommand{\uhr}{{\upharpoonright}}

\begin{document}
	\maketitle
	
\begin{abstract}
In this paper, we use algorithmic tools, effective dimension and Kolmogorov complexity, to study the fractal dimension of distance sets. We show that, for any analytic set $E\subseteq\R^2$ of Hausdorff dimension strictly greater than one, the \textit{pinned distance set} of $E$, $\Delta_x E$, has Hausdorff dimension of at least $\frac{3}{4}$, for all points $x$ outside a set of Hausdorff dimension at most one. This improves the best known bounds when the dimension of $E$ is close to one.
\end{abstract}

\section{Introduction}
The \textit{distance set} of a set $E\subseteq\R^n$ is
\begin{center}
$\Delta E = \{\|x-y\| \mid x,y\in E\}$.
\end{center}
A famous conjecture of Erd\"os, the \textit{distinct distances problem}, states that for finite sets, the size of $\Delta E$ is nearly linear in terms of the size of $E$. Guth and Katz \cite{GutKat15}, in a breakthrough paper, settled the distinct distances problem in the plane. 

Falconer posed an analogous question for the case that $E$ is infinite, known as Falconer's \textit{distance set problem}. When $E$ is infinite, we measure its size using Hausdorff dimension. Falconer's original conjecture is that, for any Borel set $E\subseteq\R^n$ with $\dim_H(E) > n/2$,  $\Delta E$ has positive measure. Even in the planar case, the distance problem is still open. However there has been substantial progress in the past few years. The current best known bound is due to Guth, Iosevich, Ou and Wang \cite{GutIosOuWang20}, who proved that if a planar set has $\dim_H(E) > 5/4$, then $\mu(\Delta E) > 0$. In fact, they showed under this hypothesis, that the pinned distance set $\Delta_x$ has positive measure, for many $x\in\R^2$ (see below).

In addition to Falconer's original question, exciting progress has been made on several variants of the problem \cite{Orponen17, Shmerkin19, KelShm19, ShmWang21}.  One interesting direction is to prove lower bounds on the Hausdorff dimension of $\Delta E$. Another is to understand the Hausdorff dimension of \textit{pinned distance sets}. If $E\subseteq\R^n$ and $x\in\R^n$, then the pinned distance set is
\begin{center}
$\Delta_x E = \{\|x-y\|\mid y\in E\}$.
\end{center}
The conjecture for pinned distance sets is that, if $\dim_H(E) > n/2$, then for ``most" points $x$, $\dim_H(\Delta_x E) = 1$. Recently, there has been exciting progress on the dimensional variants of the distance set problem. For planar pinned distance sets, Liu \cite{Liu20} showed that, if $\dim_H(E)= s\in (1, 5/4)$, then for many $x$, $\dim_H(\Delta_x E) \geq \frac{4}{3}s - \frac{2}{3}$. Shmerkin \cite{Shmerkin20} improved this bound for small values of $s$. Specifically, he showed that, for any planar Borel set with $\dim_H(E) = s \in (1, 1.04)$, $\dim_H(\Delta_x E) \geq 2/3 + 1/42 \approx 0.6904$ for many $x$. Furthermore, for the non-pinned version of the problem, Shmerkin proved that, if $\dim_H(E) \in (1, 1.06)$, then $\dim_H(\Delta E) \geq 2/3 + 2/57 \approx 0.7017$.

In this paper, we improve the bound on the dimension of pinned distance sets for small values of $\dim_H(E)$.
\begin{thm}\label{thm:maintheorem}
Let $E \subseteq \R^2$ be an analytic set with Hausdorff dimension strictly greater than one. Then, for all $x\in\R^2$ outside a set of Hausdorff dimension at most $1$,
\begin{center}
$\dim_H(\Delta_x E) \geq \frac{s}{4} + \frac{1}{2}$,
\end{center}
where $s = \dim_H(E)$.
\end{thm}

To prove this theorem, we use \textit{effective dimension}. Effective, i.e., algorithmic, dimensions were introduced \cite{Lutz03a, AthHitLutMay07} to study the amount of randomness of points in Euclidean space. The effective dimension, $\dim(x)$, of a Euclidean point $x$ is a real value which measures the asymptotic density of information of  $x$. Recent work has shown that algorithmic dimension is not only useful in effective settings, but, via the point-to-set principle, can be used to solve problems in geometric measure theory \cite{Lutz17, LutStu20, LutStu18, Stull18, Stull22a}. 

Before giving the full details, we give a brief description of the main ideas involved in the proof of Theorem \ref{thm:maintheorem}. A main component of the proof is bounding the complexity of \textit{orthogonal directions}, which may be of independent interest.

For $x\in \R^2$ and $e\in\mathcal{S}^1$, the projection of $x$ onto $e$ is $p_e x = e\cdot x$. The behavior of fractal dimensions under orthogonal projections is one of the most studied topics in geometric measure theory. Recently, orthogonal projections have been studied using effective methods \cite{LutStu18, Stull22a}. In the language of effective dimension, if $e$ is random relative to $x$, then we have tight bounds on the complexity of $p_e x$. However, for general directions, optimal bounds are unknown. In Section \ref{sec:projections}, we improve known bounds on the complexity of $p_e x$ for directions with high initial complexity (but which are not necessarily random everywhere). 

The proof of our projection result relies on analyzing the small scale structure of the complexity function $K_r(x)$. Similar ideas have been used in the classical setting  \cite{Shmerkin20, KelShm19}.  We give a more thorough overview of the proof of these results in Section \ref{sec:projections}.

In the second part of the proof of Theorem \ref{thm:maintheorem}, we prove an analogous bound for the \textit{effective} dimension of distances between \textit{points}. This proof combines the projection theorem and a naive lower bound of the complexity of distances in a bootstrapping argument. We expand on the intuition behind this proof in Section \ref{sec:effdim}.

The final part of the proof of Theorem \ref{thm:maintheorem} is a reduction to its effective counterpart, Theorem \ref{thm:mainThmEffDim}. We show that, if $E\subseteq\R^2$ is analytic and $\dim_H(E) > 1$, then for any $x$ outside a set of Hausdorff dimension at most $1$, there is a $y\in E$ so that $x,y$ satisfy the conditions necessary for the effective theorem to hold. 

This reduction relies on a result of Orponen on radial projections \cite{Orponen19DimSmooth}. Orponen's theorem has been used in much of the recent work on distance sets. We note that this result is the primary reason we require $\dim_H(E) > 1$. We combine Orponen's theorem with recent work on independence results in algorithmic information theory \cite{Stull22a} to complete the proof of the reduction.

\section{Preliminaries}\label{sec:prelim}

\subsection{Kolmogorov complexity and effective dimension}
The \emph{conditional Kolmogorov complexity} of binary string $\sigma\in\{0,1\}^*$ given a binary string $\tau\in\{0,1\}^*$ is the length of the shortest program $\pi$ that will output $\sigma$ given $\tau$ as input. Formally, the conditional Kolmogorov complexity of $\sigma$ given $\tau$ is
\[K(\sigma\mid\tau)=\min_{\pi\in\{0,1\}^*}\left\{\ell(\pi):U(\pi,\tau)=\sigma\right\}\,,\]
where $U$ is a fixed universal prefix-free Turing machine and $\ell(\pi)$ is the length of $\pi$. Any $\pi$ that achieves this minimum is said to \emph{testify} to, or be a \emph{witness} to, the value $K(\sigma\mid\tau)$. The \emph{Kolmogorov complexity} of a binary string $\sigma$ is $K(\sigma)=K(\sigma\mid\lambda)$, where $\lambda$ is the empty string.	We can easily extend these definitions to other finite data objects, e.g., vectors in $\Q^n$, via standard binary encodings. See~\cite{LiVit08} for details.

The \emph{Kolmogorov complexity} of a point $x\in\R^m$ at \emph{precision} $r\in\N$ is the length of the shortest program $\pi$ that outputs a \emph{precision-$r$} rational estimate for $x$. Formally, this is 
\[K_r(x)=\min\left\{K(p)\,:\,p\in B_{2^{-r}}(x)\cap\Q^m\right\}\,,\]
where $B_{\ve}(x)$ denotes the open ball of radius $\ve$ centered on $x$. The \emph{conditional Kolmogorov complexity} of $x$ at precision $r$ given $y\in\R^n$ at precision $s\in\R^n$ is
\[K_{r,s}(x\mid y)=\max\big\{\min\{K_r(p\mid q)\,:\,p\in B_{2^{-r}}(x)\cap\Q^m\}\,:\,q\in B_{2^{-s}}(y)\cap\Q^n\big\}\,.\]
When the precisions $r$ and $s$ are equal, we abbreviate $K_{r,r}(x\mid y)$ by $K_r(x\mid y)$. As a matter of notational convenience, if we are given a non-integral positive real as a precision parameter, we will always round up to the next integer. Thus $K_{r}(x)$ denotes $K_{\lceil r\rceil}(x)$ whenever $r\in(0,\infty)$.

A simple, but useful property of Kolmogorov complexity is the following.
\begin{lem}\label{lem:MD:3.9}
			\textup{(Case and J. Lutz~\cite{CasLut15})}
			There is a constant $c\in\N$ such that for all $n,r,s\in\N$ and $x\in\R^n$, 
			\[K_r(x)\leq K_{r+s}(x)\leq K_r(x)+K(r)+ns+a_s+c\,,\]
			where $a_s=K(s)+2\log(\lceil\frac12\log n\rceil+s+3)+(\lceil\frac12\log n\rceil+3)n+K(n)+2\log n$.
\end{lem}
We may \emph{relativize} the definitions in this section to an arbitrary oracle set $A \subseteq \N$. We will frequently consider the complexity of a point $x \in \R^n$ \emph{relative to a point} $y \in \R^m$, i.e., relative to an oracle set $A_y$ that encodes the binary expansion of $y$ is a standard way. We then write $K^y_r(x)$ for $K^{A_y}_r(x)$. Note that, for every $x\in\R^n$ and $y\in\R^m$,
\begin{equation}\label{eq:OraclesDontIncrease}
K_{s,r}(x\mid y)\geq K^y_s(x) - O(\log r) - O(\log s),
\end{equation}
for every $s, r\in\N$

One of the most useful properties of Kolmogorov complexity is that it obeys the \emph{symmetry of information}. That is, for every $\sigma, \tau \in \sigma\in\{0,1\}^*$,
\[K(\sigma, \tau) = K(\sigma) + K(\tau \mid \sigma, K(\sigma)) + O(1)\,.\]

We will need the following technical lemmas which show that versions of the symmetry of information hold for Kolmogorov complexity in $\R^n$. The first Lemma~\ref{lem:unichain} was proved in our previous work~\cite{LutStu20}.
\begin{lem}[\cite{LutStu20}]\label{lem:unichain}
	For every $m,n\in\N$, $x\in\R^m$, $y\in\R^n$, and $r,s\in\N$ with $r\geq s$,
	\begin{enumerate}
		\item[\textup{(i)}]$\displaystyle |K_r(x\mid y)+K_r(y)-K_r(x,y)\big|\leq O(\log r)+O(\log\log \|y\|)\,.$
		\item[\textup{(ii)}]$\displaystyle |K_{r,s}(x\mid x)+K_s(x)-K_r(x)|\leq O(\log r)+O(\log\log\|x\|)\,.$
	\end{enumerate}
\end{lem}
	
A consequence of Lemma \ref{lem:unichain}, is the following. 
\begin{lem}[\cite{LutStu20}]\label{lem:symmetry}
	Let $m,n\in\N$, $x\in\R^m$, $z\in\R^n$, $\ve > 0$ and $r\in\N$. If $K^x_r(z) \geq K_r(z) - O(\log r)$, then the following hold for all $s \leq r$.
	\begin{enumerate}
		\item[\textup{(i)}]$\displaystyle K^x_s(z) \geq K_s(z) - O(\log r)\,.$
		\item[\textup{(ii)}]$\displaystyle K_{s, r}(x \mid z) \geq K_s(x)- O(\log r)\,.$
	\end{enumerate}
\end{lem}

\bigskip

J. Lutz~\cite{Lutz03a} initiated the study of effective dimensions (also known as \emph{algorithmic dimensions}) by effectivizing Hausdorff dimension using betting strategies called~\emph{gales}, which generalize martingales.  Mayordomo showed that effective Hausdorff dimension can be characterized using Kolmogorov complexity~\cite{Mayordomo02}. In this paper, we use this characterization as a definition.

The \emph{effective Hausdorff dimension}  of a point $x\in\R^n$ is
\[\dim(x)=\liminf_{r\to\infty}\frac{K_r(x)}{r}\,.\]

\subsection{The Point-to-Set Principle}\label{subsec:ptsp}
The following \emph{point-to-set principles} show that The Hausdorff dimension of a set  can be characterized by the effective dimension of its individual points. The first point-to-set principle, for a restricted class of sets, was implicitly proven by J. Lutz~\cite{Lutz03a} and Hitchcock~\cite{Hitchcock05}. 

A set $E \subseteq \R^n$ is \textit{effectively compact} if the set of open rational covers of $E$ is computably enumerable. We will use the fact that every compact set is effectively compact relative to some oracle. 
\begin{thm}[\cite{Lutz03a,Hitchcock05}]\label{thm:strongPointToSetDim}
Let $E \subseteq \R^n$ and $A \subseteq \N$ be such that $E$ is effectively compact relative to $A$. Then
\[\dimH(E) = \sup\limits_{x \in E} \dim^A(x)\,.\]
\end{thm}
\begin{remark}
In fact this point-to-set principle can be generalized to $\Sigma^0_1$ subsets, but we do not need the full generality.
\end{remark}

J. Lutz and N. Lutz~\cite{LutLut18} improved this result to hold for \textit{any} set $E\subseteq\R^n$, at the cost of introducing an oracle.
\begin{thm}[Point-to-set principle~\cite{LutLut18}]\label{thm:p2s}
Let $n \in \N$ and $E \subseteq \R^n$. Then
\begin{equation*}
\dimH(E) = \adjustlimits\min_{A \subseteq \N} \sup_{x \in E} \dim^A(x).
\end{equation*}
\end{thm}

In order to apply Theorem \ref{thm:strongPointToSetDim} to the pinned distance sets, we need the following well known fact of computable analysis. We give a brief sketch of the proof for completeness. 
\begin{prop}\label{prop:effcompact}
Let $E \subseteq \R^2$ be a compact set and let $A\subseteq\N$ be an oracle  relative to which $E$ is effectively compact. Then, for every $x\in\R^2$, $\Delta_x E$ is effectively compact relative to $(x, A)$.
\end{prop}
\begin{proof}
Let $x\in\R^2$ and denote the distance function from $x$ by $f$: $f(y) = \|x-y\|$. Recall that an open set $O$ is \textit{computably open} if the set of open rational balls with rational centers contained in $O$ is computably enumerable. It is clear that for any open rational interval $(a,b)\subseteq \R$, $f^{-1}(a,b)$ is computably open, relative to $x$. 

Consider the following Turing machine $M$, which, given oracle $(x, A)$ and given any finite set rational intervals $I_1,\ldots, I_k \subseteq \R$ as input does the following. $M$ begins enumerating a set $\mathcal{B}^i = B^i_1, B^i_2,\ldots$ of open, rational, balls whose union is $f^{-1}(I_i)$. For each finite subset $S$ of $\cup_i \mathcal{B}^i$, $M$ checks if $S$ covers $E$. Note that this is possible, since $E$ is effectively compact relative to $A$. If it does, $M$ accepts. Thus $\Delta_x E$ is computably compact relative to $(x, A)$. Since $x$ was arbitrary, the conclusion follows.
\end{proof}

\subsection{Classical dimension}\label{ssection:outerMeasure}
For every $E\subseteq [0,1)^n$, define the $s$-dimensional Hausdorff content at precision $r$ by
\begin{center}
$h^s_r(E) = \inf\left\{ \sum_i d(Q_i)^s \, |\, \bigcup_i Q_i \text{ covers } E \text{ and } d(Q_i) \leq 2^{-r}\right\}$,
\end{center}
where $d(Q)$ is the diameter of ball $Q$. We define the $s$-dimensional Hausdorff measure of $E$ by
\begin{center}
$\mathcal{H}^s(E) = \lim\limits_{r\to \infty} h^s_r(E)$.
\end{center}
\begin{remark}
It is well-known that $\mathcal{H}^s$ is a metric outer measure for every $s$.
\end{remark}

The \textit{Hausdorff dimension} of a set $E$ is then defined by 
\begin{center}
$\dim_H(E) = \inf\limits_{s}\{ \mathcal{H}^s(E) = \infty \} = \sup\limits_s \{\mathcal{H}^s(E) = 0\}$.
\end{center}

We will make use of the following facts of geometric measure theory (see, e.g., \cite{Falconer14}, \cite{BisPer17}).
\begin{thm}\label{thm:compactSset}
The following are true.
\begin{enumerate}
\item Suppose $E \subseteq \R^n$ is compact and satisfies $\mathcal{H}^s(E) > 0$. Then there is a compact subset $F\subseteq E$ such that $0< \mathcal{H}^s(F) <\infty$.
\item Every analytic set $E\subseteq \R^n$ has a $\Sigma^0_2$ subset $F \subseteq E$ such that $\dim_H(F) = \dim_H(E)$.
\end{enumerate}
\end{thm}

For two outer measures $\mu$ and $\nu$, $\mu$ is said to be \textit{absolutely continuous with respect to} $\nu$, denoted $\mu \ll \nu$, if $\mu(A) = 0$ for every set $A$ for which $\nu(A) = 0$.

\subsection{Helpful lemmas}\label{subsec:primarylemmas}
In this section, we recall several lemmas which were introduced by Lutz and Stull \cite{LutStu20, LutStu18} and which will be used frequently throughout the paper.

Let $x\in\R^2$, $e\in\mathcal{S}^1$ and $t, r\in\N$. Suppose that the set
\begin{center}
$N = \{w\in B_{2^{-t}}(x)\mid p_e w = p_e x \text{ and } K_r(w)\leq K_r(x)\}$
\end{center}
was covered by $B_{2^{-r}}(x)$. Then, given a $2^{-r}$-approximation of $(p_e x, e)$ and a $2^{-t}$-approximation of $x$, we can compute (an approximation of) $x$ by simply enumerating all $w \in N$. The following lemma formalizes this intuition.
\begin{lem}[\cite{LutStu18}]\label{lem:point2}
Suppose that $z\in\R^2$, $e \in \mathcal{S}^{1}$, $r\in\N$, and $\eta, \ve\in\Q_+$ satisfy the following conditions.
\begin{itemize}
\item[\textup{(i)}]$K_r(z)\leq \eta r + \frac{\ve r}{2}$.
\item[\textup{(ii)}] For every $w \in B_{2^{-t}}(z)$ such that $p_e w = p_e z$, \[K_{r}(w)\geq \eta r + \min\{\ve r, r-s -\ve r\}\,,\]
whenever $s=-\log\|z-w\|\leq r$.
\end{itemize}
Then for every oracle set $A\subseteq\N$,
\[K_{r,r,t}^{A}(z \mid e,z) \leq K^{A}_{r,r,t}( p_e z \mid e,z) + 3\ve r + K(\ve,\eta)+O(\log r)\,.\]
\end{lem}

The next lemma shows that the precision to which we can compute $e$ given $x, w$ such that $p_e x = p_e w$ depends linearly on the distance of $x$ and $w$.
\begin{lem}[\cite{LutStu20}]\label{lem:lowerBoundOtherPoints}
Let $z \in \R^2$, $e \in S^{1}$, and $r \in \N$. Let $w \in \R^2$ such that  $p_e z = p_e w$. Then 
\begin{equation*}
    K_r(w) \geq K_t(z) + K_{r-t,r}(e\mid z) + O(\log r)\,,
\end{equation*}
where $t := -\log \|z-w\|$.
\end{lem}

We will commonly need to lower the complexity of points at specified positions. The following lemma shows that conditional complexity gives a convenient way to do this.
\begin{lem}[\cite{LutStu20}]\label{lem:oracles}
Let $z\in\R^2$, $\eta \in\Q_+$, and $r\in\N$. Then there is an oracle $D=D(r,z,\eta)$ with the following properties.
\begin{itemize}
\item[\textup{(i)}] For every $t\leq r$,
\[K^D_t(z)=\min\{\eta r,K_t(z)\}+O(\log r)\,.\]
\item[\textup{(ii)}] For every $m,t\in\N$ and $y\in\R^m$,
\[K^{D}_{t,r}(y\mid z)=K_{t,r}(y\mid z)+ O(\log r)\,,\]
and
\[K_t^{z,D}(y)=K_t^z(y)+ O(\log r)\,.\]
\item[\textup{(iii)}] If $B\subseteq\N$ satisfies $K^B_r(z) \geq K_r(z) - O(\log r)$, then \[K_r^{B,D}(z)\geq K_r^D(z) - O(\log r)\,.\]
\item[\textup{(iv)}] For every $t\in\N$, $u\in\R^n, w\in\R^m$
\[K_{r,t}(u\mid w) \leq K^D_{r,t}(u\mid w) + K_r(z) - \eta r + O(\log r)\,.\]
\end{itemize}
In particular, this oracle $D$ encodes $\sigma$, the lexicographically first time-minimizing witness to $K(z\uhr r\mid z\uhr s)$, where $s = \max\{t \leq r \, : \, K_{t-1}(z) \leq \eta r\}$.
\end{lem}

We will also need the following analog of Lemma \ref{lem:point2}. The proof of this Lemma is nearly identical to that of Lemma \ref{lem:point2}, however, we give the proof for completeness.
\begin{lem}\label{lem:pointDistance}
Suppose that $x, y\in\R^2$, $t<r\in\N$, and $\eta, \ve\in\Q_+$ satisfy the following conditions.
\begin{itemize}
\item[\textup{(i)}]$K_r(y)\leq \left(\eta +\frac{\ve}{2}\right)r$.
\item[\textup{(ii)}] For every $w \in B_{2^{-t}}(y)$ such that $\|x-y\| = \|x-w\|$, \[K_{r}(w)\geq \eta r + \min\{\ve r, r-s -\ve r\}\,,\]
where $s=-\log\|y-w\|\leq r$.
\end{itemize}
Then for every oracle set $A\subseteq\N$,
\[K_{r,t}^{A, x}(y \mid y) \leq K^{A,x}_{r,t}( \|x-y\|\mid y) + 3\ve r + K(\ve,\eta)+O(\log r)\,.\]
\end{lem}
\begin{proof}
Suppose $x,y, r, \eta, \ve$, and $A$ satisfy the hypothesis.
				
Define an oracle Turing machine $M$ that does the following given oracle $(A,x)$,  and inputs $\pi= \pi_1\pi_2\pi_3\pi_4$ and $q\in\Q^2$ such that $U(\pi_2)=s_1\in\N$, $U(\pi_3) = s_2\in\N$ and $U(\pi_4)=(\zeta, \iota) \in\Q^2$.
				
$M$ first uses $q$ and oracle access to $x$ to compute the dyadic rational $d$ such that $\vert d - \|x - q\|\vert < 2^{-s_1}$. $M$ then computes $U^{A.x}(\pi_1, q) = p \in \Q$. For every program $\sigma\in\{0,1\}^*$ with $\ell(\sigma)\leq \iota s_2 + \frac{\zeta s_2}{2}$, in parallel, $M$ simulates $U(\sigma)$. If one of the simulations halts with output $p_2 \in \Q^2 \cap B_{2^{-s_1}}(q)$ such that $\vert d - \|x-p_2\|\vert < 2^{-s_2}$, then $M^{A, e}$ halts with output $p_2$. Let $k_M$ be a constant for the description of $M$.
				
Let $\pi_1$, $\pi_2$, $\pi_3$, and $\pi_4$ testify to $K^{A,e}_{r,t}(\|x-y\| \mid y)$, $K(t)$, $K(r)$, and $K(\ve,\eta)$, respectively, and let $\pi= \pi_1\pi_2\pi_3\pi_4$. Let $q$ be a rational such that $\|y-q\| < 2^{-t}$. Let $\sigma$ be a program of length at most $(\eta +\frac{\ve}{2})r$ such that $\|p - y\| < 2^{-r}$, where $U(\sigma) = p$. Note that such a program must exist by condition (i) of our hypothesis. Note that
\begin{align*}
\vert \|x - y\| - \|x - p\|\vert &\leq \|y - p\|\\
&< 2^{-r}.
\end{align*}
Thus on input $\pi$, $M^{A, e}$ is guaranteed to halt.

Let $M^{A, e}(\pi) = p_2 \in \Q^2$. Then, by the definition of $M$, $p_2\in B_{2^{-t}}(y)$, and $\vert \|x-p_2\| - \|x-y\| \vert < 2^{1-r}$. Observation~\ref{obs:existProj} shows that there is some
\[w \in B_{\|x-y\|2^{1-r}}(p_2)\subseteq B_{2^{-t}}(y)\]
such that $\|x-y\| = \|x-w\|$. In addition, by the definition of $M$, $K(p) \leq (\eta + \frac{\ve}{2})r$. Therefore, by Lemma \ref{lem:MD:3.9},
\begin{center}
$K_r(w) \leq (\eta + \frac{\ve}{2})r + 2\log \lceil \|x-y\|\rceil + O(\log r)$.
\end{center}
Let $s = -\log\|y-w\|$. We first assume that $\ve r \leq r - s -\ve r$. Then, by condition (ii),
\begin{center}
$\eta r + \ve r \leq (\eta + \frac{\ve}{2})r + 2\log \lceil \|x-y\|\rceil + O(\log r)$,
\end{center}
a contradiction, since $r$ was chosen to be sufficiently large. Now assume that $r-s-\ve r < \ve r$. Then,
\begin{center}
$\eta r + r-s -\ve r \leq (\eta + \frac{\ve}{2})r + 2\log \lceil \|x-y\|\rceil + O(\log r)$,
\end{center}
and so $r - s \leq \frac{3\ve r}{2} + O(\log r)$. Therefore, by Lemma \ref{lem:MD:3.9},
\begin{center}
$K^{A,x}_r(y\mid w) \leq 3\ve r + O(\log r)$.
\end{center}

Hence,
\begin{align*}
K^{A,x}_{r,t}(y\mid y) &\leq \vert \pi \vert + 3\ve r + O(\log r)\\
&= K^{A,x}_{r,t}(\|x-y\|\mid y|) + 3\ve r + K(\ve, \eta) + O(\log r),
\end{align*}
and the proof is complete.
\end{proof}

\section{Projection results}\label{sec:projections}
In this section, we prove tight bounds on the Kolmogorov complexity of projections of points onto a line. For $x\in\R^2$ and $e\in\mathcal{S}^1$, the \textit{projection of $x$ onto the line in the direction of $e$} is $p_e x = e\cdot x$. The central question of this section is to understand the quantity $K_r(x\mid p_e x, e)$. That is, how difficult is it to compute $x$ if you are given its projection? This can be translated into a more geometric question. Let
\begin{center}
$N = \{w \in \R^2 \mid p_e w = p_e x, \text{ and } K_r(w)\leq K_r(x)\}$.
\end{center}
Then
\begin{center}
$K_{r}(x\mid p_e x, e) \approx r + \log \mu(N)$.
\end{center}

Lutz and Stull \cite{LutStu20} answered\footnote{Using the point-to-set principle, this gave a new, algorithmic, proof of Marstrand's projection theorem.} this question for the restricted case when $e$ is \textit{random}, relative to $x$. They showed that, under this assumption,
\begin{equation}
K_{r}(x\mid p_e x, e) \approx \min\{0, K_r(x) - r\}.
\end{equation}

Unfortunately, for the application of projection bounds to pinned distance sets, we do not have complete control over the direction $e$. Thus, we cannot guarantee that $e$ is random, and so we cannot apply this bound. However, we will have enough control over $e$ to ensure that it is random up to some precision $t$. That is, $K^x_s(e) \approx s$, for all $s \leq t$. The main theorem of this section generalizes the result in \cite{LutStu20}  to give reasonable bounds in this case.
\begin{thm}\label{thm:projectionMainTheorem}
Let  $x \in \R^2$, $e \in \mathcal{S}^1$, $\ve\in \Q^+$, $C\in\N$, and $t, r \in \N$. Suppose that $r$ is sufficiently large, and that the following hold.
\begin{enumerate}
\item[\textup{(P1)}] $\dim(x) > 1$.
\item[\textup{(P2)}] $t \geq \frac{r}{C}$.
\item[\textup{(P3)}] $K^x_s(e) \geq s - C\log s$, for all $s \leq t$. 
\end{enumerate}
Then 
\begin{center}
$K_{r}(x \, | \, p_e x, e) \leq K_r(x) - \frac{r+t}{2} + \ve r$.
\end{center}
\end{thm}
\begin{remark}
This is a generalization of \cite{LutStu20}. If $e$ is random relative to $x$, then by taking $t = r$, we can recover the theorem of Lutz and Stull. We also remark that, when $t$ is significantly smaller than $r$, it is likely that this bound can be improved.
\end{remark}

The proof of Theorem \ref{thm:projectionMainTheorem} relies on analyzing the function $s \mapsto K_s(x)$. Instead of working with this function directly, it is convenient to work with a function defined on real numbers. Let $f:\R_+ \rightarrow \R_+$ be the piece-wise linear function which agrees with $K_s(x)$ on the integers, and 
\begin{center}
$f(a) = f(\lfloor a\rfloor) + (a -\lfloor a\rfloor)(f(\lceil a \rceil) - f(\lfloor a\rfloor)) $,
\end{center}
for any non-integer $a$. Since
\begin{center}
$K_s(x) \leq K_{s+1}(x)$ 
\end{center} 
for every $s\in\N$, $f$ is non-decreasing and by Lemma \ref{lem:MD:3.9}, for every $s_1 < s_2\in\N$,
\begin{center}
$f(b) - f(a) \leq 2(b-a) + O(\log \lceil b \rceil)$.
\end{center}
Finally, we note that, for any $a\in \N$,
\[
f(a+1) - f(a) =
\begin{cases} 
0 
    & \text{  } \\
1
    & \text{  }\\
M &\text{ for some } M \geq 2.
\end{cases}
\]

The proof of Theorem \ref{thm:projectionMainTheorem} relies on extending the techniques of \cite{LutStu20} to certain nice intervals. An interval $[a,b]$ is called \textbf{\textit{teal}} if $f(b) - f(c) \leq b-c$ for every $a\leq c\leq b$. It is called \textbf{\textit{yellow}} if $f(c) - f(a) \geq c - a$, for every $a\leq c \leq b$. We denote the set of teal intervals by $T$ and the set of yellow intervals by $Y$.

We conclude this section with a brief outline of Theorem \ref{thm:projectionMainTheorem}. In Section \ref{subsec:projectionsIntervals}, we show that the techniques of \cite{LutStu20} extend naturally on yellow and teal intervals\footnote{This idea is implicit in \cite{LutStu22, Stull22b}, in the context of points on a line}. Given a sufficiently nice partition of $[0,r]$, we are able to simply sum the bounds of each subinterval (Lemma \ref{lem:boundGoodPartitionProjection}). In Section \ref{subsec:proofTheoremProjection} we construct partition of $[0,r]$ which minimizes the complexity of $K_r(x\mid p_e x, e)$, and completes the proof.

\subsection{Projection bounds on yellow and teal intervals}\label{subsec:projectionsIntervals}
The proofs of Lemmas \ref{lem:projectionsIncreasingInterval} and \ref{lem:projectionsDecreasingInterval} follow the same strategy. Since 
\begin{center}
$K_{r,t}(x\mid p_e x, e,x) \approx r +\log \mu\left( \{w \in B_{2^{-t}(x)} \mid p_e w = p_e x, \text{ and } K_r(w)\leq K_r(x)\}\right)$.
\end{center}
a bound for $K_{r,r,r,t}(x\mid p_e x, e,x)$, is, roughly, equivalent to bounding the set of points $w\in B_{2^{-t}}(x)$ such that $p_e w = p_e x$. Thus, if this set contains only points close to $x$, then $K_{r,r,r,t}(x\mid p_e x, e,x)$ is small. Lemma \ref{lem:point2} formalizes this intuition. 

Thus, we need to show that any $w$ with $p_e w = p_e x$ is close to $x$. We do this by using an \textit{incompressibility argument}. Suppose that $K^x_s(e) \approx s$, for all $s \leq r-t$. If $w$ has the same projection as $x$, then we can compute $e$ given $w$ and $x$. This is the content of Lemma \ref{lem:lowerBoundOtherPoints}, in the language of Kolmogorov complexity, which shows that
\begin{center}
$K^x_{r-t}(e) \lesssim K^x_r(w)$
\end{center}
where $t = -\log\|w-x\|$. Using symmetry of information arguments, we can show that, if $[t, r]$ is teal, then this contradicts our assumption that $e$ was random at small precisions. When $[t, r]$ is yellow, we use Lemma \ref{lem:oracles} to artificially decrease the complexity of $x$ at precision $r$. Once this is done, we are in the case where $[t, r]$ is teal.

We begin by assuming that $t$ and $r$ are rationals, and slightly weaken the yellow and teal property. 
\begin{lem}\label{lem:projectionsIncreasingInterval}
Let $x\in\R^2$, $e\in\mathcal{S}^1$, $\ve \in \Q_+$, $C\in\N$, and $t < r\in\N$. Suppose that $r$ is sufficiently large, depending on $\ve$, and that the following hold.
\begin{enumerate}
\item[\textup{(i)}] $K^x_s(e) \geq s - C\log r$, for all $s \leq r - t$.
\item[\textup{(ii)}] $K_{s,t}(x) \geq s - t - \ve r$, for all $t\leq s\leq r$.
\end{enumerate}
Then,
\begin{center}
$K_{r,r,r, t}(x\mid p_e x, e, x) \leq K_{r, t}(x\mid x) -(r - t) + 7\ve r$.
\end{center}
\end{lem}
\begin{remark}
The lower bound for a precision $r$ to be sufficiently large depends only on $\ve$, $C$ and $\|x\|$. An informal estimate is that any precision $r$ greater than $M\left(\frac{C}{\ve}\right)^2$ is sufficient, for some fixed constant $M$.
\end{remark}
\begin{proof}
Let $\eta \in \Q_+$ such that $\eta r = K_t(x) + r - t - 3\ve r$. Let $D = D(r, x, \eta)$ be the oracle of Lemma \ref{lem:oracles}.  Note that, by Lemma \ref{lem:oracles}(i),
\begin{equation}\label{eq:projectionsIncreasingIntervals1}
K^D_r(x) \leq (\eta +\frac{\ve }{2})r
\end{equation}
Following the intuitive outline given above, it suffices to show that any $w\in B_{2^{-t}}(x)$ is either very close to $x$, or $K^D_r(w) \gg K^D_r(x)$. That is, we need to show that the conditions of Lemma \ref{lem:point2} are satisfied.

Suppose that $w\in B_{2^{-t}}(x)$ such that $p_e x = p_e w$. Let $s = -\log\|w - x\|$. We first deduce that
\begin{align}
K^D_r(w) &\geq K^D_s(x) + K^D_{r-s,r}(e\mid x) - O(\log r)\tag*{[Lemma \ref{lem:lowerBoundOtherPoints}]}\\
&\geq K^D_s(x) + K_{r-s,r}(e\mid x)  - O(\log r)\tag*{[Lemma \ref{lem:oracles}(ii)]}\\
&\geq K^D_s(x) + K^x_{r-s}(e)  - O(\log r)\tag*{[(\ref{eq:OraclesDontIncrease})]}\\
&\geq K^D_s(x) + r - s - O(\log r)\tag*{[Condition (i)]}.
\end{align}
By Lemma \ref{lem:oracles}(i), $K^D_s(x) = \min\{\eta r, K_s(x)\} + O(\log r)$. If $K^D_s(x) = K_s(x) - O(\log r)$, then 
\begin{align}
K^D_r(w) &\geq K_s(x) + r - s - O(\log r)\tag*{}\\
&= K_t(x) + K_{s,t}(x) + r - s  -O(\log r)\tag*{[Lemma \ref{lem:unichain}(ii)]}\\
&\geq K_t(x) + r - t - \ve r- O(\log r)\tag*{[Condition (ii)]}\\
&= (\eta + 2\ve)r - O(\log r)\tag*{}\\
&\geq (\eta + \ve) r\tag*{}.
\end{align}
For the other case, when $K^D_s(x) = \eta r - O(\log r)$,
\begin{align*}
K^D_r(w) &\geq  \eta r + r - s - O(\log r)\\
&\geq \eta r + r- s -\ve r
\end{align*}
Therefore,
\begin{equation}\label{eq:projectionsIncreasingIntervals2}
K^D_r(w) \geq \eta r + \min\{\ve r, r - s - \ve r\}.
\end{equation}

\smallskip

Inequalities (\ref{eq:projectionsIncreasingIntervals1}) and (\ref{eq:projectionsIncreasingIntervals2}) show that the conditions of Lemma \ref{lem:point2} are satisfied relative to $D$. Applying this lemma, relative to $D$, yields
\begin{equation}\label{eq:ProjIntervalIncreasingMain}
K_{r,r,t}^{D}(x \mid e, x) \leq K^{D}_{r,r, t}( p_e x \mid e, x) + 3\ve r + K(\ve,\eta)+O(\log r).
\end{equation}
Using symmetry of information arguments, it is routine to show that inequality (\ref{eq:ProjIntervalIncreasingMain}) implies the conclusion of the present lemma. Formally, 
\begin{align}
K^D_{r,r,r, t}(x\mid p_e x, e, x) &= K^D_{r,r, t}(x \mid e, x) - K^D_{r,r, t}(p_e x\mid e, x) + O(\log r)\tag*{[Lemma \ref{lem:unichain}]}\\
&\leq  3\ve r + K(\ve,\eta) + O(\log r)\label{eq:align1}.
\end{align}
To remove the presence of oracle $D$, we use the property (iv) of Lemma \ref{lem:oracles}.
\begin{align}
K_{r, r, r,t}(x\mid p_e x, e, x) &\leq K^D_{r,r,r,t}(x\mid p_e x, e, x) \tag*{[Lemma \ref{lem:oracles}(iv)]}\\
&\;\;\;\;\;\;\;\;\;\;+ K_r(x) - \eta r + O(\log r)\tag*{}\\
&\leq 3\ve r + K(\ve,\eta)+ K_r(x) - \eta r + O(\log r)\tag*{[(\ref{eq:align1})]}\\
&= K_{r,t}(x\mid x) -(r-t) + 6\ve r + K(\ve, \eta) + O(\log r)\label{eq:align2}.
\end{align}
Finally, by the symmetry of information, and the definition of $\eta$,
\begin{align*}
K(\eta, \ve) &= K(\ve) + K(\eta\mid \ve, K(\ve))\\
&\leq K(\ve) + O(\log r).
\end{align*}
Therefore, using (\ref{eq:align2}) and the assumption that $r$ is sufficiently large,
\begin{align*}
K_{r, r, r,t}(x\mid p_e x, e, x) &\leq K_{r,t}(x\mid x) - (r-t) + 6\ve r + K(\ve) + O(\log r)\\
&\leq K_{r,t}(x\mid x) - (r-t) + 7\ve r.
\end{align*}
\end{proof}

We now bound the complexity in the case that $[t,r]$ is approximately teal. The proof is largely identical to that of Lemma \ref{lem:projectionsIncreasingInterval}. However, since $[t, r]$ is teal, we only need a small amount of information to compute $x$ from its projection.
\begin{lem}\label{lem:projectionsDecreasingInterval}
Let $x \in \R^2$, $e \in \mathcal{S}^1$, $\epsilon \in \Q^+$, $C\in\N$, and $t < r\in\N$. Suppose that $r$ is sufficiently large and that the following are satisfied.
\begin{enumerate}
\item[\textup{(i)}] $K^x_s(e) \geq s - C\log r$, for all $s \leq r - t$.
\item[\textup{(ii)}] $K_{r,s}(x) \leq r - s + \ve r$, for all $t \leq s \leq r$.
\end{enumerate}
Then 
\begin{center}
$K_{r,r,r,t}(x\mid p_e x, e, x) \leq 7\epsilon r$.
\end{center}
\end{lem}
\begin{proof}
Let $\eta \in \Q$ such that $\eta r = K_r(x) - 3\ve r$. Let $D = D(r, x, \eta)$ be the oracle of Lemma \ref{lem:oracles}. Note that, by Lemma \ref{lem:oracles}(i),
\begin{equation}\label{eq:projectionsDecreasingIntervals1}
K^D_r(x) \leq (\eta +\frac{\ve r}{2})
\end{equation}

Let $w\in B_{2^{-t}}(x)$ such that $p_e x = p_e w$. Let $s = -\log\|w - x\|$. Our goal is to show that the complexity of $w$ is sufficiently large so that we may apply Lemma \ref{lem:point2}. With that in mind, we first deduce that
\begin{align}
K^D_r(w) &\geq K^D_s(x) + K^D_{r-s,r}(e\mid x) - O(\log r)\tag*{[Lemma \ref{lem:lowerBoundOtherPoints}]}\\
&\geq K^D_s(x) + K_{r-s,r}(e\mid x)  - O(\log r)\tag*{[Lemma \ref{lem:oracles}(ii)]}\\
&\geq K^D_s(x) + r - s - O(\log r)\tag*{[Condition (i)]}.
\end{align}

By Lemma \ref{lem:oracles}(i) $K^D_s(x) = \min\{\eta r, K_s(x)\} + O(\log r)$. If $K^D_s(x) = K_s(x) - O(\log r)$, then 
\begin{align}
K^D_r(w) &\geq K_s(x) + r - s - O(\log r)\tag*{}\\
&\geq K_s(x) + K_{r,s}(x) - \ve r -O(\log r)\tag*{[Condition (ii)]}\\
&= K_r(x) - \ve r- O(\log r)\tag*{[Lemma \ref{lem:unichain}]}\\
&\geq K_r(x) - 2\ve r\tag*{}\\
&= (\eta + \ve) r\tag*{}.
\end{align}
Alternatively, if $K^D_s(x) = \eta r - O(\log r)$,
\begin{align*}
K^D_r(w) &\geq  \eta r + r - s- O(\log r)\\
&\geq \eta r + r- s -\ve r.
\end{align*}
Therefore,
\begin{equation}\label{eq:projectionsDecreasingIntervals2}
K^D_r(w) \geq \eta r + \min\{\ve r, r - s - \ve r\}.
\end{equation}

Inequalities (\ref{eq:projectionsDecreasingIntervals1}) and (\ref{eq:projectionsDecreasingIntervals2}) show that the conditions of Lemma \ref{lem:point2} are satisfied relative to $D$. Applying this yields
\begin{equation}\label{eq:ProjIntervalDecreasingMain}
K_{r,r,t}^{D}(x \mid e, x) \leq K^{D}_{r,r, t}( p_e x \mid e, x) + 3\ve r + K(\ve,\eta)+O(\log r).
\end{equation}
To complete the proof, we use symmetry of information and the properties of our oracle to show that (\ref{eq:ProjIntervalDecreasingMain}) implies the conclusion of our lemma. Formally
\begin{align}
K^D_{r,r,r, t}(x\mid p_e x, e, x) &= K^D_{r,r, t}(x \mid e, x) - K^D_{r,r, t}(p_e x\mid e, x) + O(\log r)\tag*{[Lemma \ref{lem:unichain}]}\\
&\leq  3\ve r + K(\ve,\eta) + O(\log r)\label{eq:align3}.
\end{align}

Again using the properties of our oracle $D$,
\begin{align}
K_{r, r, r,t}(x\mid p_e x, e, x) &\leq K^D_{r,r,r, t}(x\mid p_e x, e, x)\tag*{[Lemma \ref{lem:oracles}(iv)]}\\
&\;\;\;\;\;\;\;\;\;\; + K_r(x) - \eta r + O(\log r)\tag*{}\\
&\leq 3\ve r + K(\ve,\eta) +K_r(x) - \eta r + O(\log r)\tag*{[(\ref{eq:align3})]}\\
&= 6\ve r + K(\ve, \eta) + O(\log r)\label{eq:align4}.
\end{align}

Therefore, combining (\ref{eq:align4}), and the fact that 
\begin{equation}
K(\eta\mid \ve) \leq O(\log r)\tag*{},
\end{equation}
we see that
\begin{align*}
K_{r, r, r,t}(x\mid p_e x, e, x) &\leq 6\ve r + K(\ve,\eta) + O(\log r)\\
&= 6\ve r + K(\ve) + O(\log r)\\
&\leq 7\ve r,
\end{align*}
and proof is complete.
\end{proof}

With Lemmas \ref{lem:projectionsIncreasingInterval} and \ref{lem:projectionsDecreasingInterval} it is not difficult to extend to the case when $t,r$ are not necessarily rational. Additionally, the following Corollary helps prevent the error terms from cluttering the notation.
\begin{cor}\label{cor:projectionYellowTeal}
Let $x\in\R^2, e\in\mathcal{S}^1, \ve\in\Q_+$, $C\in\N$ and $t<r\in\R_+$. Suppose that $r$ is sufficiently large and $K^x_s(e) \geq s - C\log r$, for all $s\leq r-t$. Then the following hold.
\begin{enumerate}
\item If $[t,r]$ is yellow, 
\begin{center}
$K_{r,r,r,t}(x\mid p_e x, e,x) \leq K_{r,t}(x\mid x) - (r-t) + \ve r $.
\end{center}
\item If $[t,r]$ is teal, 
\begin{center}
$K_{r,r,r,t}(x\mid p_e x, e,x) \leq \ve r$.
\end{center}
\end{enumerate}
\end{cor}
\begin{proof}
We begin by noting that condition (i) of Lemmas \ref{lem:projectionsIncreasingInterval} and \ref{lem:projectionsDecreasingInterval} hold immediately. 

We assume that $r$ is large enough to satisfy Lemmas \ref{lem:projectionsIncreasingInterval} and \ref{lem:projectionsDecreasingInterval} with respect to $\ve^\prime = \ve / 10$. Let $t^\prime = \lceil t \rceil$ and $r^\prime = \lceil r \rceil$. 

First assume that $[t, r]$ is yellow. It suffices to show that condition (ii) of Lemma \ref{lem:projectionsIncreasingInterval} holds. For every $t^\prime \leq s \leq r^\prime$,
\begin{align}
K_{s,t^\prime}(x\mid x) &= K_s(x) - K_{t^\prime}(x) - O(\log t^\prime)\tag*{[Lemma \ref{lem:unichain}]}\\
&= f(s) - f(t^\prime) - O(\log t^\prime)\tag*{}\\
&\geq f(s) - f(t) - O(\log t^\prime)\tag*{}\\
&\geq s - t - \ve^\prime r\tag*{},
\end{align} 
and the claim follows by Lemma \ref{lem:MD:3.9}.

Now assume that $[t, r]$ is teal. It suffices to show that condition (ii) of Lemma \ref{lem:projectionsDecreasingInterval} holds. For every $t^\prime \leq s \leq r^\prime$,
\begin{align*}
K_{r^\prime,s}(x\mid x) &= K_{r^\prime}(x) - K_{s}(x) + O(\log r^\prime)\\
&= f(r^\prime) - f(s) + O(\log r^\prime)\\
&\leq f(r) - f(s) + O(\log r^\prime)\\
&\leq r - s + \ve^\prime r.
\end{align*} 
Thus both conditions (ii) of Lemma \ref{lem:projectionsDecreasingInterval} hold. Applying it, and using Lemma \ref{lem:MD:3.9} concludes the proof.
\end{proof}

\subsection{Admissible partitions}
Corollary \ref{cor:projectionYellowTeal} gives tight bounds on yellow and teal intervals. Thus, if we had a partition of $[0,r]$ consisting solely of yellow and teal intervals, we could sum the bounds on each interval to get a bound on $K_r(x\mid p_e x, e)$. This is easily done via repeated applications of the symmetry of information. In order to avoid the error terms from piling up, we require the partition to have a constant number of intervals.

We say that a partition $\mathcal{P}=\{[a_i, a_{i+1}]\}_{i=0}^k$ of closed intervals with disjoint interiors is \textbf{\textit{$(M,r,t)$-admissible}} if $[0, r] = \cup_i [a_i, a_{i+1}]$, and it satisfies the following conditions.
\begin{itemize}
\item[\textup{(A1)}] $k \leq M$,
\item[\textup{(A2)}] $[a_i, a_{i+1}]$ is either yellow or teal,
\item[\textup{(A3)}] $a_{i+1} \leq a_i + t$.
\end{itemize}

Given an admissible partition containing $O(1)$ intervals, we use the symmetry of information and Corollary \ref{cor:projectionYellowTeal} to bound the complexity $K_r(x\mid p_e x, e)$.
\begin{lem}\label{lem:boundGoodPartitionProjection}
Suppose that $x \in \R^2$, $e \in \mathcal{S}^1$, $\ve\in \Q^+$, $C\in\N$, $t, r \in \N$ satisfy (P1)-(P3). If $\mathcal{P} = \{[a_i, a_{i+1}]\}_{i=0}^k$ is an  $(3C,r,t)$-admissible partition, and $r$ is sufficiently large, then
\begin{align*}
K_{r}(x \mid p_e x, e) &\leq \ve r + \sum\limits_{i\in \textbf{Bad}} K_{a_{i+1}, a_{i}}(x \mid x) - (a_{i+1} - a_i),
\end{align*}
where
\begin{center}
\textbf{Bad} $=\{i\leq k\mid [a_i, a_{i+1}] \notin T\}$.
\end{center}
\end{lem}
\begin{proof}
To begin, we assume that, if $[a_i, a_{i+1}]\in\mathcal{P}$ and $a_{i+1} \geq \log r$, then $a_{i+1}$ is large enough so that Corollary \ref{cor:projectionYellowTeal} holds with respect to the choices of $x$, $e$, $C$, $a_i$, $a_{i+1}$ and $\ve^\prime = \ve/6C$. Since $r$ is assumed to be sufficiently large, this assumption is valid. 

With this in mind, we modify $\mathcal{P}$ by removing all intervals with right endpoints less than $\log r$, and adding the interval $[0, \log r]$. We relabel the intervals so that $a_1 = \log r$, and  $a_i >\log r$ for all $1< i\leq k$.

For any teal interval $[a_i,a_{i+1}]\in \mathcal{P}$ such that  $a_{i}\geq \log r$, by Corollary \ref{cor:projectionYellowTeal}, with respect to $\ve^\prime$,
\begin{equation}\label{eq:boundGoodPartitionProjection1}
K_{a_{i+1},a_{i+1},a_{i+1}, a_i}(x \mid p_e x, e,x) \leq \ve^\prime a_{i+1}.
\end{equation}
Similarly, for any yellow interval $[a_i,a_{i+1}]\in \mathcal{P}$ such that  $a_{i}\geq \log r$, by Corollary \ref{cor:projectionYellowTeal}, with respect to $\ve^\prime$,
\begin{equation}\label{eq:boundGoodPartitionProjection2}
K_{a_{i+1},a_{i+1},a_{i+1}, a_i}(x \mid p_e x, e)  \leq K_{a_{i+1},a_i}(x\mid x) - (a_{i+1}-a_i) + \ve^\prime a_{i+1}.
\end{equation}

By repeated applications of the symmetry of information, and inequalities (\ref{eq:boundGoodPartitionProjection1}) and (\ref{eq:boundGoodPartitionProjection2}),
\begin{align}
K_{r-t}(x \mid p_e x, e) &\leq O(\log r) + \sum\limits_{i = 0}^k K_{a_{i+1}, a_{i+1}, a_{i+1},a_{i}}(x \mid p_e x, e, x)\tag*{[(A1)]}\\
&\leq O(\log r) + K_{\log r}(x) \tag*{}\\
&\;\;\;\;\;\;\;\;\;\;\;\;\;\;\;\;+ \sum\limits_{i = 1}^k K_{a_{i+1}, a_{i+1}, a_{i+1},a_{i}}(x \mid p_e x, e, x)\tag*{}\\
&\leq \ve^\prime r +\sum\limits_{I_i\in T} K_{a_{i+1}, a_{i+1}, a_{i+1},a_{i}}(x \mid p_e x, e,x)\tag*{[r is large]} \\
&\;\;\;\;\;\;\;\;\;\;\;\;\;\;\;\; + \sum\limits_{i\in \textbf{Bad}} K_{a_{i+1}, a_{i+1}, a_{i+1},a_{i}}(x \mid p_e x, e,x)\tag*{}\\
&\leq \ve r+\sum\limits_{i\in \textbf{Bad}} K_{a_{i+1}, a_{i}}(x\mid x)- (a_{i+1} - a_i)\tag*{[(\ref{eq:boundGoodPartitionProjection1}), (\ref{eq:boundGoodPartitionProjection2})]},
\end{align}
and the proof is complete.
\end{proof}

We now show that admissible partitions do exist for any subinterval $[a,b]$ of $[0,r]$.
\begin{lem}\label{lem:goodPartitionProjection}
Let $x\in\R^2$, $r, C\in\N$ and $\frac{r}{C}\leq t < r$. For any $0\leq a < b \leq r$, there is an $(3C,r,t)$-admissible partition of $[a,b]$.
\end{lem}
\begin{proof}
Define the following procedure for every $a\leq c \leq b$.

\bigskip

\noindent \textbf{Procedure}$(c,b)$:
\begin{itemize}
\item If $c = b$, halt.
\item Let 
\begin{center}
$d = \max\limits_{d\leq c+t, b} \{[c,d] \in Y\cup T\}$.
\end{center}
\item Add $[c,d]$ to $\mathcal{P}$.
\item Call Procedure$(d, b)$.
\end{itemize}
It suffices to show that the partition $\mathcal{P}$ produced by Procedure$(a,b)$ is $3C$-admissible. 

We begin by showing that the algorithm is well defined and eventually halts. Let $a\leq c < b$. By the definition of $f$, it is clear that the set
\begin{center}
$\{d\mid d\leq c+t, \,[c,d] \in Y\cup T\}$
\end{center}
is non-empty and closed, and so the algorithm is well defined. Moreover, $d \geq \max\{\lceil c\rceil, \lfloor c\rfloor+1\}$, and so the algorithm halts. Thus, the partition $\mathcal{P}$ produced by Procedure$(a, b)$ satisfies conditions (A2) and (A3) of good partitions.

To complete the proof, we now prove that $\mathcal{P} = \{[a_i, a_{i+1}]\}_{i=0}^k$ satisfies (A1), i.e., $\mathcal{P}$ contains at most $3C+1$ intervals. It suffices to show that $a_{i+2} > a_i + t$, for every $i < k-1$. Let $[a_i, a_{i+1}]\in\mathcal{P}$. If it is yellow, then we must have $a_{i+1} = a_i + t$. So, assume that $[a_{i}, a_{i+1}]$ is teal. If $[a_{i+1}, a_{i+2}]$ is also teal, then by the definition of our algorithm, $a_{i+2} > a_i + t$. If $[a_{i+1}, a_{i+2}]$ is yellow, then $a_{i+2} = a_{i+1} +t$, and the proof is complete.
\end{proof}

\subsection{Proof of Theorem \ref{thm:projectionMainTheorem}}\label{subsec:proofTheoremProjection}
Unfortunately, Lemmas \ref{lem:boundGoodPartitionProjection}  and \ref{lem:goodPartitionProjection} are not enough to imply Theorem \ref{thm:projectionMainTheorem}. The partition constructed in Lemma \ref{lem:goodPartitionProjection} is too naive to give the bounds required. Intuitively, we need our partition to minimize the contribution of the ``bad" intervals (the intervals which are not teal). We accomplish this by giving a finer classification of types intervals.

We introduce three more colors of intervals: red, blue and green. An interval $[a,b]$ is called \textbf{\textit{red}} if $f$ is strictly increasing on $[a,b]$. An interval $[a,b]$ is called \textbf{\textit{blue}} if $f$ is constant on $[a,b]$. An interval $[a,b]$ is \textbf{\textit{green}} if it is both yellow and teal, and its length is at most $t$. That is,
\[
[a, b] \text{ is }
\begin{cases} 
\text{red} 
    & \text{ if }f(d) - f(c) > d-c \text{ for every } a \leq c < d \leq b \\
\text{blue}
    & \text{ if }f(a) = f(c)  \text{ for every } a \leq c \leq b\\
\text{green} &\text{ if } [a,b] \in Y, \, b-a \leq t, \text{ and } f(b) -f(a) = b-a.
\end{cases}
\]
From the definition of $f$, it is clear that every $s\in [0,r]$ is covered by a red, blue or green interval. We also make the following observation.
\begin{obs}\label{obs:redbluegreen}
If $[a, b]$ is red and $[b, c]$ is blue, then $b$ is contained in the interior of a green interval.
\end{obs}

Let $G$ denote the set of all green intervals. We remove any interval $I_1\in G$ such that $I_1 \subseteq I_2$ for some $I_2\in G$. For every $[a,b], [c,d] \in G$ such that $a < c < b < d$, remove $[c,d]$ from $G$, and add $[b, d]$ to $G$. Thus, after this process is complete, any two intervals of $G$ have disjoint interiors. 

\medskip

\noindent \textbf{Construction of} $\hat{P} = P(x,r,t)$: We now construct a partition $\hat{P} = P(x,r,t)$ of $[0, r]$ of intervals with disjoint interiors, each of which is either red, blue or green. To begin, we add all intervals in $G$ to $\hat{P}$. Let $[b, c]\subseteq [0,r]$ be an interval whose interior is disjoint from $\hat{P}$ and such that $[a,b], [c,d]\in \hat{P}$, for some $a, d$. Then if
\[
[b, c] \text{ is }
\begin{cases} 
\text{red} 
    & \text{ we add } [b,c] \text{ to } \hat{P} \\
\text{blue}
    & \text{ we add } [b,c]  \text{ to } \hat{P}\\
[b,s] \cup [s,c] &\text{ we add } [b,s], [s,c] \text{ to } \hat{P}.
\end{cases}
\]
Note that in the third case, $[b,s]$ is blue and $[s,c]$ is red. It is clear that the intervals of $\hat{P}$ have disjoint interiors, and their union is $[0, r]$.

A \textbf{\textit{red-green-blue}} sequence in $\hat{P}$ is a sequence $I_0,\ldots, I_{n+1}$ of consecutive intervals in $\hat{P}$ such that $I_0$ is red, $I_1,\ldots, I_n$ are green, and $I_{n+1}$ is blue.
\begin{obs}\label{obs:redgreenblue}
Suppose $I_0,\ldots, I_{n+1}$ is a a red-green-blue sequence in $\hat{P}$. Then the total length of $I_1,\ldots, I_n$ is at least $t$.
\end{obs}
\begin{proof}
We first note that, if $[a,b]$, $[b,c]\in\hat{P}$ are green, then $c > a + t$. Hence, it suffices to prove the claim in the case of a red-green-blue sequence $[a,b]$, $[b,c]$, $[c,d]$. For all sufficiently small $\delta > 0$, since $[a,b]$ is red and $[c,d]$ is blue,
\begin{align*}
f(c+\delta) - f(a) &= f(c) - f(a)\\
&\geq c - b + 2(b - a)\\
&> c + \delta - a.
\end{align*}
On the other hand, $f(c+\delta) - f(b) < c+\delta - b$. Therefore, there is some $s \in (a, b)$ such that $[s, c+\delta]$ is yellow and $f(c+\delta) - f(s) = c+\delta - s$. Hence, if $c+\delta \leq s + t$, $[s, c+\delta]$ is green. By our construction of $\hat{P}$, $[s, c+\delta]$ is not green, and so we must have that $c = b+t$.
\end{proof}

The last lemma we need to prove Theorem \ref{thm:projectionMainTheorem} gives the existence of an admissible partition in the case that there are no red-green-blue sequences in $\hat{P}$. The main idea is that, since $\dim(x) > 1$, no red-green-blue sequences implies that there are no blue intervals in $\hat{P}$ (after some small precision). Thus we can construct an admissible partition $\mathcal{P}$ containing only yellow intervals.
\begin{lem}\label{lem:admissiblenoredgreenblue}
Let  $x \in \R^2$, $C\in\N$, and $t, r \in \N$. Suppose $\dim(x) > 1$, $t\geq \frac{r}{C}$ and $r$ is sufficiently large. If there are no red-green-blue sequences in $\hat{P}$, then there is a $3C$-admissible partition $\mathcal{P}$ of $[0,r]$ such that 
\begin{equation*}
\sum\limits_{I_i \in \mathcal{P}-Y} a_{i+1} - a_i \leq s,
\end{equation*}
for some constant $s$ depending only on $x$. 
\end{lem}
\begin{proof}
Let $[a,b]$ be the first red interval of $\hat{P}$. Then, since there are no red-green-blue sequences, there are no blue intervals in $\hat{P}$ which intersect $[b, r]$. Since $\dim(x) > 1$, there is a constant $s^\prime$ such that $f(s) > s$, for all $s \geq s^\prime$. Therefore, $a < s^\prime$.

We now construct an admissible partition $\mathcal{P}$ of $[0,r]$ so that $[a,r]$ is covered by only yellow intervals in $\mathcal{P}$. Define the following procedure, for any $a\leq c < r$.

\bigskip

\noindent \textbf{Procedure}$(c,r)$:
\begin{itemize}
\item If $c = r$, halt.
\item Let 
\begin{center}
$d = \max\limits_{d\leq c+t, r} \{[c,d] \in Y \text{ and } d\notin int(I) \text{ for any green } I\in \hat{P}\}$.
\end{center}
\item Add $[c,d]$ to $\mathcal{P}$.
\item Call Procedure$(d, b)$.
\end{itemize}

Let $a\leq c < r$ be any real not contained in the interior of a green interval of $\hat{P}$. Then, since $[a,r]$ is covered exclusively by red and green intervals of $\hat{P}$,
\begin{center}
$\{d\mid [c,d] \in Y  \text{ and } d\notin int(I) \text{ for any green } I\in \hat{P}\}$
\end{center}
is non-empty and closed. Hence, the procedure is well defined. Moreover, $d \geq \max\{\lfloor c \rfloor + 1, \lceil c \rceil\}$, and so Procedure$(a,r)$ eventually halts. Therefore, by construction $\mathcal{P}$ satisfies (A2) and (A3). 

Finally, we show that $\mathcal{P} = \{[a_i, a_{i+1}]\}^k_{i=0}$ satisfies (A1). Let $i < k$. Then, since $[a_i, a_{i+1}]$ and $[a_{i+1},a_{i+2}]$ are yellow,  $[a_i, a_{i+2}]$ is yellow. By the definition of our algorithm, we must have $a_{i+2} > a_i + t$, and so $\mathcal{P}$ is $2C$-admissible.

Finally, we extend $\mathcal{P}$ to cover $[0, a]$ with yellow and teal intervals using Lemma \ref{lem:goodPartitionProjection}. Since $a < s^\prime$, we see that
\begin{equation*}
\sum\limits_{I_i \in \mathcal{P}-Y} a_{i+1} - a_i \leq a < s^\prime,
\end{equation*}
and the proof is complete.
\end{proof}

\begin{T2}
Let  $x \in \R^2$, $e \in \mathcal{S}^1$, $\ve\in \Q^+$, $C\in\N$, and $t, r \in \N$. Suppose that $r$ is sufficiently large, and that the following hold.
\begin{enumerate}
\item[\textup{(P1)}] $\dim(x) > 1$.
\item[\textup{(P2)}] $t \geq \frac{r}{C}$.
\item[\textup{(P3)}] $K^x_s(e) \geq s - C\log s$, for all $s \leq t$. 
\end{enumerate}
Then 
\begin{center}
$K_{r}(x \, | \, p_e x, e) \leq K_r(x) - \frac{r+t}{2} + \ve r$.
\end{center}
\end{T2}
\begin{proof}
Assume $x, e, r, t$ and $\ve$ satisfy the hypothesis. Let $\hat{P} = P(x,r,t)$ be the partition of $[0,r]$ into red, blue and green intervals. There are two cases to consider. 

First assume that there are no red-green-blue sequences in $\hat{P}$. Let $\mathcal{P}$ be the $3C$-admissible partition of Lemma \ref{lem:admissiblenoredgreenblue}. Thus
\begin{equation*}
\sum\limits_{I_i \in \mathcal{P}-Y} a_{i+1} - a_i \leq s,
\end{equation*}
for some constant $s$, and so
\begin{equation}
B := \sum\limits_{I_i \in \mathcal{P}\cap Y} a_{i+1} - a_i \geq r - \frac{\ve r}{2}.
\end{equation}

By repeated applications of the symmetry of information, and Lemma \ref{lem:boundGoodPartitionProjection} with respect to $\ve / 4$,
\begin{align}
K_r(x) &\geq \sum\limits_{I_i \in \mathcal{P}\cap Y} K_{a_{i+1}, a_i}(x\mid x) - O(\log r)\tag*{[Lemma \ref{lem:unichain}]}\\
&\geq -\frac{\ve r}{4} + \sum\limits_{I_i \in \mathcal{P}\cap Y} K_{a_{i+1}, a_i}(x\mid x)\tag*{[$\mathcal{P}$ admissible]}\\
&\geq K_r(x\mid p_e x, e) + B -\frac{\ve r}{2}\tag*{[Lemma \ref{lem:boundGoodPartitionProjection}]}\\
&\geq K_r(x\mid p_e, x, e) + r - \ve r.\label{eq:projectionMainThm4}
\end{align}
Rearranging (\ref{eq:projectionMainThm4}) , we see that
\begin{align*}
K_r(x\mid p_e x, e) &\leq K_r(x) - r + \ve r\\
&\leq K_r(x) - \frac{r+t}{2} + \ve r,
\end{align*}
and the proof is complete for this case.

We now consider the case that there is at least one red-green-blue sequence in $\hat{P}$. In particular, by Observation \ref{obs:redgreenblue}, there is a sequence of consecutive green intervals $I_1,\ldots, I_n\in \hat{P}$ with total length at least $t$. By our construction of $\hat{P}$, there must be consecutive\footnote{Technically, there could be only one green interval of length exactly $t$. In this case, abusing notation, we take the second interval to empty.} green intervals $[a,b]$, $[b,c] \in \hat{P}$ such that $c \geq a + t$. Let $\mathcal{P}_1$ be a $3C$-admissible partition of $[0,a]$, and $\mathcal{P}_2$ be a $3C$-admissible partition of $[c, r]$, guaranteed by Lemma \ref{lem:goodPartitionProjection}. Let 
\begin{center}
$\mathcal{P} = \mathcal{P}_1 \cup [a,b] \cup [b, c] \cup \mathcal{P}_2$.
\end{center}
Note that $\mathcal{P}$ is a $10C$-admissible partition. 

Let 
\begin{equation}\label{eq:lengthOfTealIntervals}
L = \sum\limits_{I_i \in \mathcal{P}\cap G} a_{i+1} - a_i
\end{equation}
be the total length of the green intervals in $\mathcal{P}$. Note that $L \geq t$. Let 
\begin{equation}
B = \sum\limits_{i\in \textbf{Bad}} a_{i+1}-a_i
\end{equation}
be the total length of the bad (non-teal) intervals in $\mathcal{P}$. 

We first prove that 
\begin{equation}\label{eq:projectionMainThm1}
K_r(x\mid p_e x, e) \leq \ve r + \min\{K_r(x) - B - t, B\}
\end{equation}
Since $x$ is an element of $\R^2$, by Lemma \ref{lem:MD:3.9},
\begin{center}
$K_{a_{i+1}, a_i}(x\mid x) \leq 2(a_{i+1} - a_i) + O(\log r)$.
\end{center}
Therefore, by Lemma \ref{lem:boundGoodPartitionProjection}, with respect to $\ve / 4$,
\begin{equation}\label{eq:projectionMainThm2}
K_r(x\mid p_e x, e) \leq \frac{\ve r}{2} + B.
\end{equation}
By repeated applications of the symmetry of information,
\begin{align}
K_r(x) &\geq -\frac{\ve r}{2} + \sum\limits_{I_i \in \mathcal{P}\cap T} K_{a_{i+1}, a_i}(x\mid x) + \sum\limits_{i\in \textbf{Bad}} K_{a_{i+1}, a_i}(x\mid x)\tag*{}\\
&\geq -\frac{\ve r}{2} + \sum\limits_{I_i \in \mathcal{P}\cap G} K_{a_{i+1}, a_i}(x\mid x) + \sum\limits_{i\in \textbf{Bad}} K_{a_{i+1}, a_i}(x\mid x)\tag*{}\\
&= -\frac{\ve r}{2} +L + \sum\limits_{i\in \textbf{Bad}} K_{a_{i+1}, a_i}(x\mid x)\tag*{}\\
&\geq t + K_r(x\mid p_e x, e) + B -\ve r.\label{eq:projectionMainThm3}
\end{align}
Combining (\ref{eq:projectionMainThm2}) and (\ref{eq:projectionMainThm3}) proves inequality (\ref{eq:projectionMainThm1}). 

By inequality (\ref{eq:projectionMainThm1}), if
\begin{center}
$B \leq K_r(x) - \frac{r+t}{2}$,
\end{center}
we are done, so we assume otherwise. Applying (\ref{eq:projectionMainThm1}) again, this implies that
\begin{align}
K_r(x\mid p_e x,e) &\leq \ve r+ K_r(x) - t - B\tag*{[(\ref{eq:projectionMainThm1})]}\\
&< \ve r+ K_r(x) - t - K_r(x) + \frac{r+t}{2}\tag*{}\\
&= \ve r+ \frac{r-t}{2}\tag*{}\\
&< \ve r+ K_r(x) - \frac{r+t}{2}\tag*{[(P1)]},
\end{align}
and the proof is complete.
\end{proof}

\section{Complexity of distances}\label{sec:complexityDistances}
In this section, we prove a lower bound on the complexity, $K^x_r(\|x-y\|)$, of the distance between two points.
\begin{thm}\label{thm:DistanceLowerBoundComplexity}
Suppose that $x, y \in \R^2$, $e = \frac{x - y}{\|x - y\|}$ and $C\in\N$ satisfy the following for every $r\in\N$.
\begin{itemize}
\item[\textup{(D1)}] $K_r(y) > r - C\log r$.
\item[\textup{(D2)}] $K^x_r(y) > K_r(y) - C\log r$.
\item[\textup{(D3)}] $K_r(e\mid y) =  r - o(r)$.
\end{itemize}
Then, for every $\ve\in\Q_+$ and all sufficiently large $r\in\N$,
\begin{center}
$K^x_{r}(\|x - y\|) \geq \frac{K_r(y)}{2} - \ve r$.
\end{center}
\end{thm}

Although this bound tight\footnote{This result is the best possible, for example, when $\dim(y) = 2$.} in general, we can do better when the complexity of $y$ is significantly less than $2r$. Indeed, Theorem \ref{thm:mainThmEffDim}, shows that we can conclude a bound of $\frac{3r}{4}$. The key point of Theorem \ref{thm:DistanceLowerBoundComplexity} is that it relates the complexity of $\|x-y\|$ to the complexity of $\|y\|$  \textit{at every precision}. We will use this in the proof of Theorem \ref{thm:mainThmEffDim} via a bootstrapping argument. 

Before getting into the details, we first give intuition behind the proof of Theorem \ref{thm:DistanceLowerBoundComplexity}. The proof is very similar to that of our projection theorem (Section \ref{sec:projections}). The connect is due to the fact that, if $z\in\R^2$ such that $\|z - x\| = \|y-x\|$, then (Observation \ref{obs:geometricObs})
\begin{center}
$\vert p_e y - p_e w \vert \lesssim \|y-z\|^2$.
\end{center}
Since $e$ is essentially random relative to $y$ (condition (D3)), we have techniques to bound the set of all such $z$.

The proof of Theorem \ref{thm:projectionMainTheorem} relies on analyzing the function $s \mapsto K_s(x)$. Instead of working with this function directly, it is convenient to work with a function defined on real numbers. Let $f:\R_+ \rightarrow \R_+$ be the piece-wise linear function which agrees with $K_s(y)$ on the integers, and 
\begin{center}
$f(a) = f(\lfloor a\rfloor) + (a -\lfloor a\rfloor)(f(\lceil a \rceil) - f(\lfloor a\rfloor)) $,
\end{center}
for any non-integer $a$. Since
\begin{center}
$K_s(y) \leq K_{s+1}(y)$ 
\end{center} 
for every $s\in\N$, $f$ is non-decreasing and by Lemma \ref{lem:MD:3.9}, for every $s_1 < s_2\in\N$,
\begin{center}
$f(b) - f(a) \leq 2(b-a) + O(\log \lceil b \rceil)$.
\end{center}
Finally, we note that, for any $a\in \N$,
\[
f(a+1) - f(a) =
\begin{cases} 
0 
    & \text{  } \\
1
    & \text{  }\\
M &\text{ for some } M \geq 2.
\end{cases}
\]

An interval $[a,b]$ is called \textbf{\textit{teal}} if $f(b) - f(c) \leq b-c$ for every $a\leq c\leq b$. It is called \textbf{\textit{yellow}} if $f(c) - f(a) \geq c - a$, for every $a\leq c \leq b$. We denote the set of teal intervals by $T$ and the set of yellow intervals by $Y$.

The overall strategy of Theorem \ref{thm:DistanceLowerBoundComplexity} is nearly identical to that of Theorem \ref{thm:projectionMainTheorem}. Specifically, in Section \ref{subsec:distYellowTeal}, we give bounds on yellow and teal intervals. Given a sufficiently nice partition of $[0,r]$ we can sum the bounds on individual intervals to conclude a bound on the complexity $K^x_r(y\mid \|x-y\|)$. Finally, we show the existence of a good partition which allows us to conclude the bound of Theorem \ref{thm:DistanceLowerBoundComplexity}.

\subsection{Complexity of distances on yellow and teal intervals}\label{subsec:distYellowTeal}
The proofs of Lemmas \ref{lem:distancesIncreasingIntervals} and \ref{lem:distancesDecreasingIntervals} follow the same strategy as \ref{lem:projectionsIncreasingInterval} and \ref{lem:projectionsDecreasingInterval}. Thus, the intuition given in Section \ref{subsec:projectionsIntervals} can be used here. The main difference is that we now allow our intervals to be of the form $[\frac{r}{2}, r]$, instead of $[t, r]$. We begin with a brief description of the proofs. Since 
\begin{center}
$K^x_{r,t}(y \mid\|x-y\|,y) \approx r -\log \mu\left(N\right)$,
\end{center}
where
\begin{center}
$N=\{w \in B_{2^{-t}(y)} \mid \|x-y\|=\|x-w\|, \text{ and } K_r(w)\leq K_r(y)\}$,
\end{center}
a bound for $K^x_{r,t}(y \mid\|x-y\|,y)$, is equivalent to bounding the set of points $w\in B_{2^{-t}}(y)$ such that $\|x-y\|=\|x-w\|$. Thus, if this set contains only points close to $y$, then $K^x_{r,t}(y \mid\|x-y\|,y)$ is small. Lemma \ref{lem:pointDistance} formalizes this intuition (and is the analog of Lemma \ref{lem:point2} for distances).

As in Section \ref{subsec:projectionsIntervals}, we bound the set of $w$ using an incompressibility argument. To do so, we apply the fact that, if $w\in\R^2$ such that $\|w - x\| = \|y-x\|$, then (Observation \ref{obs:geometricObs})
\begin{center}
$\vert p_e y = p_e w \vert \lesssim \|y-w\|^2$.
\end{center}
Thus, if $w\in B_{2^{-\frac{r}{2}}}(y)$, then $p_e w \approx p_e y$. We can then use the method described for projections to conclude that we can compute $e$ given $w$ and $y$ (Lemma \ref{lem:lowerBoundOtherPoints}). Therefore, if $[\frac{r}{2}, r]$ is teal, this contradicts our assumption that $e$ is random. When $[\frac{r}{2}, r]$ is yellow, we use Lemma \ref{lem:oracles} to artificially decrease the complexity of $y$ at precision $r$. Once this is done, we are in the case where $[t, r]$ is teal.
\begin{lem}\label{lem:distancesIncreasingIntervals}
Let $x,y \in \R^2$, $e = \frac{x-y}{\|x-y\|}$, $\ve\in\Q_+$, and $t < r\leq 2t \in\N$. Suppose $r$ is sufficiently large and the following are satisfied.
\begin{enumerate}
\item[\textup{(i)}] $K_{s,r}(e \mid y) \geq s - \frac{\ve r}{2}$, for all $s \leq r - t$
\item[\textup{(ii)}] $K_{s,t}(y) \geq  s-t - \ve r$ for all $t \leq s \leq r$.  
\end{enumerate}
Then 
\begin{center}
$K^x_{r,r,t}(y\mid \|x - y\|, y) \leq K_{r,t}(y\mid y) - (r - t) + 7\ve r$.
\end{center}
\end{lem}
\begin{proof}
Let $\eta \in \Q$ such that $\eta r = K_t(y) + r - t - 3\ve r$. Let $D = D(r, y, \eta)$ be the oracle of Lemma \ref{lem:oracles}. Note that, by Lemma \ref{lem:oracles}(i), $K^D_r(y) \leq \eta r + \frac{\ve r}{2}$.

Suppose that $w\in B_{2^{-t}}(y)$ such that $\|x - w\| = \|x - y\|$. Let $s = -\log\|w-y\|$. Our goal is to show a sufficiently large lower bound on the complexity of $w$, so that we may apply Lemma \ref{lem:pointDistance}. We first note that, by Observation \ref{obs:geometricObs},
\begin{center}
$\|p_{e} y - p_{e} w\| < 2^{-2t+c}\leq 2^{-r+c}$,
\end{center}
where $c = \log\|x-y\|$. Therefore, by Lemma \ref{lem:lowerBoundOtherPoints},
\begin{align}
K^D_r(w) &\geq K^D_{r-c}(z) \tag*{}\\
&\geq K^D_s(y) + K^D_{r-c-s,r-c}(e\mid y) - O(\log r)\tag*{[Lemma \ref{lem:lowerBoundOtherPoints}]}\\
&\geq K^D_s(y) + K^D_{r-s,r}(e\mid y) - O(\log r)\tag*{[Lemma \ref{lem:MD:3.9}]}\\
&\geq K^D_s(y) + K_{r-s,r}(e\mid y)  - O(\log r)\tag*{[Lemma \ref{lem:oracles}(ii)]}\\
&\geq K^D_s(y) + r - s -\frac{\ve r}{2} - O(\log r)\tag*{[Condition (i)]}.
\end{align}
To continue, we bound the complexity of $y$, relative to oracle $D$. By Lemma \ref{lem:oracles}(i), there are two cases to consider. If $K^D_s(y) = K_s(y) - O(\log r)$, then 
\begin{align}
K^D_r(w) &\geq K_s(y) + r - s -\frac{\ve r}{2}- O(\log r)\tag*{}\\
&= K_t(y) + K_{s,t}(y) + r - s - \frac{\ve r}{2} -O(\log r)\tag*{[Lemma \ref{lem:unichain}(ii)]}\\
&\geq K_t(y) + r - t - \frac{3\ve r}{2}- O(\log r)\tag*{[Condition (ii)]}\\
&= K_t(y) + r - t - 2\ve r\tag*{}\\
&\geq (\eta + \ve) r\tag*{}.
\end{align}
For the other case, when $K^D_s(y) = \eta r - O(\log r)$,
\begin{align*}
K^D_r(w) &\geq  \eta r + r - s -\frac{\ve r}{2} - O(\log r)\\
&\geq \eta r + r- s -\ve r
\end{align*}
Thus the conditions of Lemma \ref{lem:pointDistance} are satisfied. Applying this, relative to $D$, yields
\begin{equation}\label{eq:distanceIntervalIncreasingMain}
K_{r,t}^{D,x}(y \mid y) \leq K^{D,x}_{r, t}( \|x-y\| \mid y) + 3\ve r + K(\ve,\eta)+O(\log r).
\end{equation}
Using symmetry of information arguments, it is not difficult to show that inequality (\ref{eq:distanceIntervalIncreasingMain}) implies the conclusion of the present lemma. Formally, 
\begin{align}
K^{D,x}_{r,r, t}(y\mid \|x-y\|, y) &= K^{D,x}_{r,r, t}(y\mid y) \tag*{[Lemma \ref{lem:unichain}]} \\
&\,\,\,\,\,\,\,\,\,\,\,\,\,\,\,\,\,\,\,\,\,- K^{D,x}_{r, t}(\|x-y\|\mid y) + O(\log r)\tag*{}\\
&\leq  3\ve r + K(\ve,\eta) + O(\log r)\label{eq:distalign1}.
\end{align}

To remove the presence of oracle $D$, we use the properties of the oracle (Lemma \ref{lem:oracles}) and symmetry of information. Formally,
\begin{align}
K^{x}_{r,r,t}(y\mid \|x-y\|, y) &\leq K^{D,x}_{r,r, t}(y\mid \|x-y\|, y)\tag*{}\\
&\;\;\;\;\;\;\;\;\; + K_r(y) - \eta r + O(\log r)\tag*{[Lemma \ref{lem:oracles}(iv)]}\\
&\leq 3\ve r + K(\ve,\eta) + K_r(y) - \eta r + O(\log r)\tag*{[(\ref{eq:distalign1})]}\\
&= K_{r,t}(y\mid y) + 6\ve r+ K(\ve,\eta) + O(\log r).\label{eq:distalign2}
\end{align}
Finally, note that 
\begin{equation}
K(\eta\mid \ve) \leq O(\log r),
\end{equation}
and so
\begin{align*}
K^{x}_{r,r,t}(y\mid \|x-y\|, y) &\leq K_{r,t}(y\mid y) - (r-t)  + 6\ve r + K(\ve) + O(\log r)\\
&\leq K_{r,t}(y\mid y) - (r-t)  + 7\ve r,
\end{align*}
and proof is complete.
\end{proof}

Using essentially the same strategy as before, we can bound the complexity on approximately teal intervals.
\begin{lem}\label{lem:distancesDecreasingIntervals}
Let $x,y \in \R^2$, $e = \frac{x-y}{\|x-y\|}$, $\ve\in\Q_+$,  and $t < r\leq 2t \in\N$. Suppose that $r$ is sufficiently large and the following are satisfied.
\begin{enumerate}
\item[\textup{(i)}] $K_{s,r}(e \mid y) \geq s - \frac{\ve r}{2}$, for all $s \leq r - t$
\item[\textup{(ii)}]  $K_{r,s}(y) \leq r - s +\ve r$ for all $t \leq s \leq r$.
\end{enumerate}
Then 
\begin{center}
$K^x_{r,r,t}(y\mid \|x-y\|, y) \leq 7\epsilon r $.
\end{center}
\end{lem}
\begin{proof}

Let $\eta \in \Q$ such that $\eta r = K_r(y) - 3\ve r$. Let $D = D(r, y, \eta)$ be the oracle of Lemma \ref{lem:oracles}. 

Suppose that $w\in\R^2$ such that $\|x - w\| = \|x - y\|$ and $w\in B_{2^{-t}}(y)$. Let $s = -\log\|w-y\|$. Our goal is to show a sufficiently large lower bound on the complexity of $w$, so that we may apply Lemma \ref{lem:pointDistance}. We first note that, by Observation \ref{obs:geometricObs},
\begin{center}
$\|p_{e} y - p_{e} w\| < 2^{-2t+c}\leq 2^{-r+c}$,
\end{center}
where $c = \log \|x-y\|$. Therefore, by Lemma \ref{lem:lowerBoundOtherPoints},
\begin{align}
K^D_r(w) &\geq K^D_{r-c}(w)\tag*{}\\
&\geq K^D_s(y) + K^D_{r-c-s,r-c}(e\mid y) - O(\log r)\tag*{[Lemma \ref{lem:lowerBoundOtherPoints}]}\\
&\geq K^D_s(y) + K^D_{r-s,r}(e\mid y) - O(\log r)\tag*{[Lemma \ref{lem:MD:3.9}]}\\
&\geq K^D_s(y) + K_{r-s,r}(e\mid y)  - O(\log r)\tag*{[Lemma \ref{lem:oracles}(ii)]}\\
&\geq K^D_s(y) + r - s -\frac{\ve r}{2} - O(\log r)\tag*{[Condition (i)]}.
\end{align}
To continue, we bound on the complexity of $y$, relative to oracle $D$. By Lemma \ref{lem:oracles}(i), there are two cases to consider. If $K^D_s(y) = K_s(y) - O(\log r)$, then 
\begin{align}
K^D_r(w) &\geq K_s(y) + r - s -\frac{\ve r}{2}- O(\log r)\tag*{}\\
&\geq K_s(y) + K_{r,s}(y) - \frac{3\ve r}{2} -O(\log r)\tag*{[Condition (ii)]}\\
&= K_r(y) - \frac{3\ve r}{2} -O(\log r)\tag*{[Lemma \ref{lem:unichain}]}\\
&= K_r(y) - 2\ve r\tag*{}\\
&= (\eta + \ve)r \tag*{}.
\end{align}
For the other case, when $K^D_s(y) = \eta r - O(\log r)$,
\begin{align*}
K^D_r(w) &\geq  \eta r + r - s -\frac{\ve r}{2} - O(\log r)\\
&\geq \eta r + r- s -\ve r
\end{align*}
Thus the conditions of Lemma \ref{lem:pointDistance} are satisfied. Applying this, relative to $D$, yields
\begin{equation}\label{eq:distanceIntervalDecreasingMain}
K_{r,t}^{D,x}(y \mid y) \leq K^{D,x}_{r, t}( \|x-y\| \mid y) + 3\ve r + K(\ve,\eta)+O(\log r).
\end{equation}
Using symmetry of information arguments, it is not difficult to show that inequality (\ref{eq:distanceIntervalDecreasingMain}) implies the conclusion of the present lemma. Formally, 
\begin{align}
K^{D,x}_{r,r, t}(y\mid \|x-y\|, y) &= K^{D,x}_{r,r,r, t}(y, \|x-y\|\mid \|x-y\|, y) \tag*{} \\
&\,\,\,\,\,\,\,\,\,\,\,\,\,\,\,\,\,\,\,\,\,- K^{D,x}_{r, t}(\|x-y\|\mid y) + O(\log r)\tag*{}\\
&= K^{D,x}_{r,r, t}(y \mid \|x-y\|, y)- K^{D,x}_{r, t}(\|x-y\|\mid y) + O(\log r)\tag*{}\\
&\leq K^{D,x}_{r, t}(y \mid y)- K^{D,x}_{r, t}(\|x-y\|\mid y) + O(\log r)\tag*{}\\
&\leq  3\ve r + K(\ve,\eta) + O(\log r)\label{eq:distalign3}.
\end{align}

To remove the presence of oracle $D$, we use the properties of the oracle (Lemma \ref{lem:oracles}) and symmetry of information. Formally,
\begin{align}
K^{x}_{r,r,t}(y\mid \|x-y\|, y) &\leq K^{D,x}_{r,r, t}(y\mid \|x-y\|, y)\tag*{}\\
&\;\;\;\;\;\;\;\;\; + K_r(y) - \eta r + O(\log r)\tag*{[Lemma \ref{lem:oracles}(iv)]}\\
&\leq 3\ve r + K(\ve,\eta) + K_r(y) - \eta r + O(\log r)\tag*{[(\ref{eq:distalign1}]}\\
&= 6\ve r+ K(\ve,\eta) + O(\log r)\tag*{}.
\end{align}
Finally, note that 
\begin{equation*}
K(\eta\mid \ve) \leq O(\log r),
\end{equation*}
and so
\begin{align*}
K^{x}_{r,r,t}(y\mid \|x-y\|, y) &\leq  6\ve r + K(\ve) + O(\log r)\\
&\leq  7\ve r,
\end{align*}
and proof is complete.
\end{proof}

\begin{cor}\label{cor:distancesYellowTeal}
Suppose that $x, y \in \R^2$ and $e = \frac{x - y}{\|x - y\|}$ satisfy (D1)-(D3) for every $r\in\N$. Then for every $\ve\in\Q_+$ and all sufficiently large $r\in\N$, the following hold.
\begin{enumerate}
\item If $[t,r]$ is yellow and $t\leq r\leq 2t$
\begin{center}
$K^x_{r,r,t}(y\mid \|x-y\|, y) \leq K_{r,t}(y\mid y) - (r-t) + \ve r $.
\end{center}
\item If $[t,r]$ is teal,  and $t\leq r\leq 2t$,
\begin{center}
$K^x_{r,r,t}(y\mid \|x-y\|, y) \leq \ve r $.
\end{center}
\end{enumerate}
\end{cor}
\begin{proof}
We assume that $r$ is large enough to satisfy Lemmas \ref{lem:distancesIncreasingIntervals} and \ref{lem:distancesDecreasingIntervals} with respect to $\ve^\prime = \ve / 10$. 

Let $t\leq r \leq 2t$. Let $t^\prime = \lceil t \rceil$ and $r^\prime = \lceil r \rceil$. Then, by (D3), condition (i) of Lemmas \ref{lem:distancesIncreasingIntervals} and \ref{lem:distancesDecreasingIntervals} hold.

Assume that $[t, r]$. It suffices to show that that Lemma \ref{lem:distancesIncreasingIntervals}(ii) holds. For every $t^\prime \leq s \leq r^\prime$,
\begin{align*}
K_{s,t^\prime}(y\mid y) &= K_s(y) - K_{t^\prime}(y) - O(\log t^\prime)\\
&= f(s) - f(t^\prime) - O(\log t^\prime)\\
&\geq f(s) - f(t) - O(\log t^\prime)\\
&\geq s - t - \ve^\prime r,
\end{align*} 
and the claim follows.

Now assume that $[t, r]$ is teal. It suffices to show that that Lemma \ref{lem:distancesDecreasingIntervals}(ii) holds. For every $t^\prime \leq s \leq r^\prime$,
\begin{align*}
K_{r^\prime,s}(y\mid y) &= K_{r^\prime}(y) - K_{s}(y) + O(\log r^\prime)\\
&= f(r^\prime) - f(s) + O(\log r^\prime)\\
&\leq f(r) - f(s) + O(\log r^\prime)\\
&\leq r - s + \ve^\prime r,
\end{align*} 
and the proof is complete.
\end{proof}

\subsection{Proof of Theorem \ref{thm:DistanceLowerBoundComplexity}}
As in the case of projections (Section \ref{sec:projections}), the proof of Theorem \ref{thm:DistanceLowerBoundComplexity} relies on the existence of a ``nice" partition of $[1,r]$. Fortunately, since we only need a weak lower bound for conclusion of Theorem \ref{thm:DistanceLowerBoundComplexity}, our partition does not have to be optimized.

We say that a partition $\mathcal{P} = \{[a_i, a_{i+1}]\}_{i=0}^k$ of intervals with disjoint interiors is \textbf{\textit{good}} if $[1, r] = \cup_i [a_i, a_{i+1}]$ and it satisfies the following conditions.
\begin{itemize}
\item[\textup{(G1)}] $[a_i, a_{i+1}]$ is either yellow or teal,
\item[\textup{(G2)}] $a_{i+1} \leq 2a_i$, for every $i$ and
\item[\textup{(G3)}] $a_{i+2} > a_{i}$ for every $i < k$.
\end{itemize}
Note that (G3) ensures that the errors do not pile up.

The following lemma uses repeated applications of the symmetry of information to write $K^x_r(y\mid \|x-y\|)$ as a sum of its complexity on the intervals of a partition. The conclusion then follows via Corollary \ref{cor:distancesYellowTeal}. 
\begin{lem}\label{lem:boundGoodPartitionDistance}
Let $\mathcal{P} = \{[a_i, a_{i+1}]\}_{i=0}^k$ be a good partition. Then
\begin{align*}
K^x_{r}(y \mid\|x-y\|) &\leq \ve r + \sum\limits_{i\in \textbf{Bad}} K_{a_{i+1}, a_{i}}(y \mid y) - (a_{i+1} - a_i),
\end{align*}
where
\begin{center}
\textbf{Bad} $=\{i\leq k\mid [a_i, a_{i+1}] \notin T\}$.
\end{center}
\end{lem}
\begin{proof}
To begin, we assume that, if $[a_i, a_{i+1}]\in\mathcal{P}$ and $a_{i+1} \geq \log r$, then $a_{i+1}$ is large enough so that Corollary \ref{cor:distancesYellowTeal} holds with respect to $\ve^\prime = \ve/8$. Since $r$ is assumed to be sufficiently large, this assumption is valid. Note that, for 
\begin{equation}
K_{\log r}(y) \leq O(\log r).
\end{equation}
With this in mind, we modify $\mathcal{P}$ by removing all intervals with right endpoints less than $\log r$, and adding the interval $[1, \log r]$. We relabel the intervals so that $a_1 = \log r$, and  $a_i >\log r$ for all $1< i\leq k$.

Let $[a_i,a_{i+1}]\in \mathcal{P}$ be a teal interval such that  $a_{i+1}\geq \log r$. Then by Corollary \ref{cor:distancesYellowTeal}, with respect to $\ve^\prime$,
\begin{equation}\label{eq:boundGoodPartitionDistance1}
K^x_{a_{i+1},a_{i+1}, a_i}(y \mid \|x-y\|, y) \leq \ve^\prime a_{i+1}.
\end{equation}
Similarly, let $[a_i,a_{i+1}]\in \mathcal{P}$ be a yellow interval such that  $a_{i+1}\geq \log r$. Then by Corollary \ref{cor:distancesYellowTeal} with respect to $\ve^\prime$,
\begin{equation}\label{eq:boundGoodPartitionDistance2}
K^x_{a_{i+1},a_{i+1}, a_i}(y \mid \|x-y\|, y)  \leq K_{b,a}(y\mid y) - (a_{i+1}-a_i) + \ve^\prime a_{i+1}.
\end{equation}

By repeated applications of the symmetry of information, and inequalities (\ref{eq:boundGoodPartitionDistance1}) and (\ref{eq:boundGoodPartitionDistance2}),
\begin{align*}
K^x_{r}(y \mid \|x-y\|) &\leq \sum\limits_{i = 0}^k K^x_{a_{i+1}, a_{i+1},a_{i}}(y \mid \|x-y\|, y) + O(\log r)\\
&\leq K_{\log r}(y) + \sum\limits_{i = 1}^k K^x_{a_{i+1}, a_{i+1}, a_{i}}(y \mid \|x-y\|, y) + O(\log r)\\
&\leq \ve^\prime r +\sum\limits_{I_i\in T} K^x_{a_{i+1}, a_{i+1},a_{i}}(y \mid \|x-y\|, y) \\
&\;\;\;\;\;\;\;\;\;\;\;\;\; + \sum\limits_{i\in \textbf{Bad}} K^x_{a_{i+1}, a_{i+1},a_{i}}(y \mid \|x-y\|, y)\\
&\leq \ve^\prime r +\sum\limits_{I_i\in T} \ve^\prime a_{i+1} \\
&\;\;\;\;\;\;\;\;\;\; + \sum\limits_{i\in \textbf{Bad}} K_{a_{i+1}, a_{i}}(y\mid y)- (a_{i+1} - a_i) + \ve^\prime a_{i+1}\\
&\leq \ve r+\sum\limits_{i\in \textbf{Bad}} K_{a_{i+1}, a_{i}}(y\mid y)- (a_{i+1} - a_i),
\end{align*}
and the conclusion follows.
\end{proof}

We now show the existence of a good partition by construction.
\begin{lem}\label{lem:existenceOfGoodPartitionDistance}
For every $x\in \R^2$ and every $r\in\N$, there is a good partition of $[1, r]$.
\end{lem}
\begin{proof}
We define the following recursive procedure which constructs $\mathcal{P}$. Given a real $a< r$, the procedure does the following.

\bigskip

\noindent \textbf{Procedure}$(a, r)$: 
\begin{itemize}
\item Let 
\begin{center}
$b = \max\limits_{d \leq 2a,r} \{[a,d] \in Y\cup T\}$.
\end{center}
\item Add $[a,b]$ to $\mathcal{P}$
\item If $b = r$, stop the procedure.
\item Otherwise, call Procedure$(b,r)$.
\end{itemize}

\bigskip

To begin, we note that, for any $a< r$, the set
\begin{center}
$\{d \leq 2a, r \mid [a,d]\in Y \cup T\}$
\end{center}
is non-empty and closed. Thus $b$ is well defined and the interval $[a,b]$ added to $\mathcal{P}$ is teal or yellow. Moreover, $b \geq \max\{\lceil a \rceil, \lfloor a \rfloor + 1\}$, and so Procedure$(1,r)$ halts. Therefore $\mathcal{P}$ satisfies (G1) and (G2). To complete the proof, we show that $\mathcal{P}$ satisfies (G3), i.e., that $a_{i+2} > 2a_i$, for every $i < k$.

Let $i < k$. First assume that $[a_{i}, a_{i+1}]$ is yellow and $a_{i+1} < 2a_i$. Then by the definition of Procedure, we must have $f(a_{i+1}+\delta) - f(a_i) < a_{i+1} +\delta - a_i$ for all sufficiently small $\delta > 0$. Therefore $f(a_{i+1}) - f(a_i) = a_{i+1} - a_i$. But then $[a_i, a_{i+1} +\delta]$ is teal, a contradiction. Thus, if $[a_i,a_{i+1}]$ is yellow, then $a_{i+1} = 2a_i$. 

Now assume that $[a_i, a_{i+1}]$ is teal. Then either $a_{i+1} = 2a_i$, or $[a_{i+1}, a_{i+2}]$ is yellow. In either case, we see that $a_{i+2} > 2a_i$, and the proof is complete.
\end{proof}

With Lemma \ref{lem:existenceOfGoodPartitionDistance}, we are now able to prove the main result of this section. 
\begin{T3}
Suppose that $x, y \in \R^2$ and $e = \frac{x - y}{\|x - y\|}$ and $C\in\N$ satisfy the following for every $r\in\N$.
\begin{itemize}
\item[\textup{(D1)}] $K_r(y) > r - C\log r$.
\item[\textup{(D2)}] $K^x_r(y) > K_r(y) - C\log r$.
\item[\textup{(D3)}] $K_r(e\mid y) =  r - o(r)$.
\end{itemize}
Then, for every $\ve\in\Q_+$ and all sufficiently large $r\in\N$,
\begin{center}
$K^x_{r}(\|x - y\|) \geq \frac{K_r(y)}{2} - \ve r$.
\end{center}
\end{T3}
\begin{proof}
Let $\mathcal{P}$ be a good partition of $[0,r]$ guaranteed by Lemma \ref{lem:goodPartitionProjection}. Let 
\begin{equation}\label{eq:lengthBadIntervalsDistance}
B = \sum\limits_{i\in \textbf{Bad}} a_{i+1}-a_i
\end{equation}
be the total length of the bad (non-teal) intervals in $\mathcal{P}$. 

Let $\ve^\prime = \ve/4$. We first prove that 
\begin{equation}\label{eq:distanceComplexityMainThm1}
K^x_r(y\mid \|x-y\|) \leq 2\ve^\prime r + \min\{K_r(x) - B, B\}
\end{equation}
Since $y$ is an element of $\R^2$, by Lemma \ref{lem:MD:3.9},
\begin{center}
$K_{a_{i+1}, a_i}(y\mid y) \leq 2(a_{i+1} - a_i) + O(\log a_i)$.
\end{center}
Therefore, by Lemma \ref{lem:boundGoodPartitionDistance}, with respect to $\ve^\prime$,
\begin{equation}\label{eq:distanceComplexityMainThm2}
K^x_r(y \mid \|x-y\|) \leq 2\ve^\prime r + B.
\end{equation}
By repeated applications of the symmetry of information,
\begin{align}
K_r(y) &\geq \sum\limits_{I_i \in \mathcal{P}} K_{a_{i+1}, a_i}(y\mid y) - O(\log a_{i+1}) \tag*{}\\
&\geq -\ve^\prime r + \sum\limits_{i\in \textbf{Bad}} K_{a_{i+1}, a_i}(y\mid y)\tag*{}\\
&\geq  K^x_r(y\mid \|x-y\|) + B - 2\ve^\prime r.\label{eq:distanceComplexityMainThm3}
\end{align}
Combining (\ref{eq:distanceComplexityMainThm2}) and (\ref{eq:distanceComplexityMainThm3}) proves inequality (\ref{eq:distanceComplexityMainThm1}).

To finish the proof, by (\ref{eq:distanceComplexityMainThm1}), we can conclude that
\begin{center}
$K^x_r(y\mid \|x-y\|) \leq \frac{K_r(y)}{2} + 2\ve^\prime r$.
\end{center}
By the symmetry of information, this implies that
\begin{align}
K^x_r(\|x-y\|) &= K^x_r(y) - K^x_r(y\mid \|x-y\|) - O(\log r)\tag*{[Lemma \ref{lem:unichain}]}\\
&\geq K_r(y) - \frac{K_r(y)}{2} - 2\ve^\prime r-O(\log r)\tag*{[(D2)]}\\
&\geq \frac{K_r(y)}{2} - \ve r\tag*{},
\end{align}
and the proof is complete.
\end{proof}

\section{Effective dimension of distances}\label{sec:effdim}
In this section, we prove the point-wise analog of the main theorem of this paper. That is, we prove the following. 
\begin{thm}\label{thm:mainThmEffDim}
Suppose that $x, y\in\R^2$, $e = \frac{y-x}{\|y-x\|}$, $d > 1$ and $A\subseteq\N$ satisfy the following.
\begin{itemize}
\item[\textup{(C1)}] $\dim^{A}(x) > 1$
\item[\textup{(C2)}] $K^{x,A}_r(e) = r - O(\log r)$ for all $r$.
\item[\textup{(C3)}] $\dim^A(y) \geq d$.
\item[\textup{(C4)}] $K^{x,A}_r(y) \geq K^{A}_r(y) - O(\log r)$ for all sufficiently large $r$. 
\item[\textup{(C5)}] $K^{A}_r(e\mid y) = r - o(r)$ for all $r$.
\end{itemize}
Then 
\begin{center}
$\dim^{x,A}(\|x-y\|) \geq \frac{d}{4} + \frac{1}{2}$.
\end{center}
\end{thm}

Before giving the details, we briefly describe the overall strategy of the proof. We first note that the conditions (C1)-(C5) ensure that the conditions of Theorem \ref{thm:projectionMainTheorem} and Theorem \ref{thm:DistanceLowerBoundComplexity} hold. 

In Theorem \ref{thm:DistanceLowerBoundComplexity}, we gave a naive lower bound on the complexity $K^x_r(\|x-y\|)$, namely, that
\begin{equation}\label{eq:effDimEq1}
K^x_r(\|x-y\|) \gtrsim \frac{K_r(y)}{2}.
\end{equation}
The key property is that it bounds the complexity of $\|x-y\|$ as a function of the complexity of $y$ \textit{at every precision}. 

The bound (\ref{eq:effDimEq1}) is not strong enough, in general, to imply Theorem \ref{thm:mainThmEffDim} by itself. However, we can leverage (\ref{eq:effDimEq1})) to improve the bound on $K^x_r(\|x-y\|)$  by combining (\ref{eq:effDimEq1})) with Theorem \ref{thm:projectionMainTheorem}. Informally, at precision $r$, we judiciously choose a precision $t < r$. Using (\ref{eq:effDimEq1})), we have
\begin{equation}\label{eq:effDimEq2}
K^x_t(\|x-y\|) \gtrsim \frac{K_t(y)}{2}.
\end{equation}
We use a separate argument (described below) which applies Theorem \ref{thm:projectionMainTheorem} to lower bound the complexity of the remaining bits of $\|x-y\|$, i.e., to bound $K^x_{r,t}(\|x-y\|)$. By choosing $t$ to be the precision at which $K_t(y) = \frac{r}{2}$, we are able to conclude that 
\begin{equation}\label{eq:effDimEq3}
K^x_{r,t}(\|x-y\|) \geq \frac{r}{2}.
\end{equation}
With this choice of $t$, (\ref{eq:effDimEq2}) implies that $K^x_t(\|x-y\|) \geq \frac{r}{4}$. Combining this with (\ref{eq:effDimEq3}) yields the desired bound.

We now briefly describe the argument using Theorem \ref{thm:projectionMainTheorem}. Informally, we prove that there are no points $z$ close to $y$ such that $\|x-z\| = \|x-y\|$. To do this, we use an \textit{incompressibility argument}. That is, we prove if there were such a $z$ then the complexity of $x$ would be smaller than it actually is. To perform the incompressibility argument, we show that $x$ is computable given $y, z$ and few additional bits.

This relies on the following observation. Suppose that $\|x-z\| = \|x-y\|$, and let $w$ be the midpoint $y$ and $z$. Then $p_{e^\prime} w = p_{e^\prime} x$, where $e^\prime = \frac{z-y}{\|z-y\|}$. Thus, to compute $x$ given $y, z$, we only need the information of $K(x\mid p_{e^\prime} x, e^\prime)$. Bounding the complexity of $K(x\mid p_{e^\prime} x, e^\prime)$ is exactly the content of Theorem \ref{thm:projectionMainTheorem}.

We now formalize this intuition. We begin with the argument that $x$ is computable given $y$ and $z$.
\begin{lem}\label{lem:lowerBoundOtherPointDistance}
Let $x, y\in \R^2$ and $r\in \N$. Let $z\in \R^2$ such that $\|x-y\| = \|x-z\|$. Then for every $B\subseteq \N$,
\begin{center}
$K^B_{r-t,r}(x\mid y) \leq K^B_r(z\mid y) + K_{r-t}(x\mid p_{e^\prime}, e^\prime) + O(\log r)$,
\end{center}
where $e^\prime = \frac{y-z}{\|y-z\|}$ and $t = -\log \|y-z\|$.
\end{lem}
\begin{remark}
In the proof of Theorem \ref{thm:mainThmEffDim}, we will use the following, equivalent, inequality.
\begin{center}
$K^B_r(z) \geq K^B_t(y) + K^B_{r-t, r}(x\mid y) - K_{r-t}(x\mid p_{e^\prime} x, e^\prime) - O(\log r)$,
\end{center}
\end{remark}
\begin{proof}
Let $x,y,z\in\R^2$ and $B\subseteq\N$ such that $\|x-y\| = \|x-z\|$. Let $w = \frac{z+y}{2}$, and $e^\prime = \frac{y-z}{\|y-z\|}$. Note that $p_{e^\prime} w = p_{e^\prime} x$.

\medskip

It is clear that
\begin{align}
K^{B}_{r-t, r}(x\mid y) &\leq K^{B}_{r}(z,w \mid y)+ K^B_{r-t, r}((p_{e^\prime} x, e^\prime) \mid (y, z,w)) \label{eq:lowerBoundDist1}\\
&\;\;\;\;\;\;\;\;\;\;\;\;\;\;\;\;+ K_{r-t}(x \mid p_{e^\prime} x, e^\prime) + O(\log r)\tag*{}
\end{align}

Since a $2^{-r}$ approximation of $w$ is computable from $2^{-r}$ approximations of $y$ and $z$, 
\begin{center}
$K^{B}_{r}(w\mid y, z) \leq O(\log r)$.
\end{center}
Thus, by the symmetry of information,
\begin{align}
K^B_r(z, w \mid y) &= K^B_r(z\mid y) + K^B_r(w\mid y,z)+ O(\log r)\tag*{}\\
&\leq K^B_r(z\mid y) + O(\log r)\label{eq:lowerBoundDist2}.
\end{align}

Let $e_2 = (e^\prime)^\bot$. Then, since $p_{e_2} y = p_{e_2} z$, $K_{r-t,r}(e_2\mid y,z) \leq O(\log r)$ (see \cite{LutStu18}). Since $e^\prime$ is computable from $e_2$, this implies that
\begin{equation}\label{eq:lowerBoundDist3}
K^B_{r-t, r}(e^\prime \mid y, z) \leq O(\log r)
\end{equation}

Finally, since $p_{e^\prime}$ is a linear function, and $p_{e^\prime} x = p_{e^\prime} w$, 
\begin{equation}\label{eq:lowerBoundDist4}
K^{B}_{r-t}(p_{e^\prime} x\mid w, e^\prime)\leq O(\log r).
\end{equation}

Combining (\ref{eq:lowerBoundDist3}) and (\ref{eq:lowerBoundDist4}) yields
\begin{align}
K^B_{r-t, r}((p_{e^\prime} x, e^\prime) \mid (y, z,w))  &= K^B_{r-t, r}(e^\prime \mid y, z,w) \tag{[Lemma \ref{lem:unichain}]}\\
&\;\;\;\;\;\;\;\; + K^{B}_{r-t, r, r-t}(p_{e^\prime} x\mid (y,z,w), e^\prime)\tag*{}\\
&\leq O(\log r)\label{eq:lowerBoundDist5}.
\end{align}

Therefore, combining (\ref{eq:lowerBoundDist1}), (\ref{eq:lowerBoundDist2}), and (\ref{eq:lowerBoundDist5}) we see that
\begin{equation}\label{eq:lowerBoundDist6}
K^{B}_{r-t, r}(x\mid y) \leq K^B_r(z \mid y) + K_{r-t}(x \mid p_{e^\prime} x, e^\prime) + O(\log r).
\end{equation}

We now show that this is indeed equivalent to the inequality stated in the remark. Since $\|y-z\| < 2^{1-t}$,
\begin{align*}
K^B_r(z\mid y) &\leq K^B_{r, t-1}(z \mid z) +O(\log r)\\
&= K^B_r(z) - K^B_{t-1}(z) +O(\log r)\\
&= K^B_r(z) - K^B_{t-1}(y) +O(\log r)\\
&= K^B_r(z) - K^B_{t}(y) +O(\log r).
\end{align*}
Combining this with (\ref{eq:lowerBoundDist6}), we see that
\begin{center}
$K^B_r(z) \geq K^{B}_{r-t, r}(x\mid y) + K^B_t(y) - K_{r-t}(x \mid p_{e^\prime} x, e^\prime) - O(\log r)$,
\end{center}
and the proof is complete.
\end{proof}

We are now able to prove the main result of this section.
\begin{T4}
Suppose that $x, y\in\R^2$, $e = \frac{y-x}{\|y-x\|}$, $A\subseteq\N$ and $d > 1$ satisfy the following.
\begin{itemize}
\item[\textup{(C1)}] $\dim^{A}(x) > 1$
\item[\textup{(C2)}] $K^{x,A}_r(e) = r - O(\log r)$ for all $r$.
\item[\textup{(C3)}] $\dim^A(y) \geq d$.
\item[\textup{(C4)}] $K^{x,A}_r(y) \geq K^{A}_r(y) - O(\log r)$ for all sufficiently large $r$. 
\item[\textup{(C5)}] $K^{A}_r(e\mid y) = r - o(r)$ for all $r$.
\end{itemize}
Then 
\begin{center}
$\dim^{x,A}(\|x-y\|) \geq \frac{d}{4} + \frac{1}{2}$.
\end{center}
\end{T4}
\begin{proof}
Let $\ve \in \Q_+$. Let $r\in\N$ be sufficiently large. Let $t\in\N$ be the precision
\begin{center}
$t = \max\{s < r \mid K^A_{s}(y) \leq \frac{dr}{2}\}$.
\end{center}
By (C3) such a $t$ exists. Moreover, by (C4) and Lemma \ref{lem:MD:3.9},
\begin{equation}\label{eq:boundComplexityAtT}
K^{A}_t(y) = \frac{dr}{2}  - O(\log r).
\end{equation}
Finally, since $K^A_s(y) \leq 2s + O(\log s)$, we see that $t \geq \frac{r}{4} - O(\log r)$. 

It is immediate from (C3), (C4) and (C5) that the conditions of Theorem \ref{thm:DistanceLowerBoundComplexity}, relative to $A$, are satisfied for $x,y, C, \ve / 2$ and $t$. Therefore by Theorem \ref{thm:DistanceLowerBoundComplexity} and equation  (\ref{eq:boundComplexityAtT}),
\begin{equation}\label{eq:lowerBoundDistanceAtT}
K^{A,x}_{t}(\|x-y\|) \geq \frac{dr}{4} - \ve r.
\end{equation}

\bigskip

\noindent We now show that 
\begin{equation}\label{eq:mainThmEffDimEq1}
K^{A,x}_{r,t}(\|x-y\|\mid y) \geq \frac{r}{2} - 4\ve r.
\end{equation}
To prove this, we will prove sufficient lower bounds on points $z\in\R^2$ whose distance from $x$ is equal to $\|x-y\|$, allowing us to apply Lemma \ref{lem:pointDistance}. 

Let $\eta\in\Q$ be a rational such that $\eta r \leq \frac{(d + 1)r}{2} - 4\ve r$. Let $D = (r, y, \eta)$ be the oracle of Lemma \ref{lem:oracles} relative to $A$. Let $z \in \R^2$  such that $z\in B_{2^{-t}}(y)$, and $\|x-y\| = \|x-z\|$. Let $s = -\log \|y-z\|$. There are two cases to consider. 

\bigskip

\noindent \textit{Case 1: $s \geq \frac{r}{2} - \log r$.} By Observation \ref{obs:geometricObs}, 
\begin{center}
$\|p_e y - p_e z\| < r^2 2^{-r}$.
\end{center}
Let $r^\prime = r - 2\log r$. Then, by Lemma \ref{lem:lowerBoundOtherPoints},
\begin{align}
K^{A,D}_r(z) &\geq K^{A,D}_{r^\prime}(z)\tag*{}\\
&\geq K^{A,D}_s(y) + K^{A,D}_{r^\prime-s,r^\prime}(e\mid y) - O(\log r)\tag*{[Lemma \ref{lem:lowerBoundOtherPoints}]}\\
&= K^{A,D}_s(y) + K^{A,D}_{r-s,r}(e\mid y) - O(\log r)\tag*{[Lemma \ref{lem:MD:3.9}]}\\
&\geq K^{A,D}_s(y) + K^A_{r-s,r}(e\mid y)  - O(\log r)\tag*{[Lemma \ref{lem:oracles}(ii)]}\\
&\geq K^{A,D}_s(y) + r - s -\frac{\ve r}{2} - O(\log r)\tag*{[(C5)]}.
\end{align}
By Lemma \ref{lem:oracles}(i),
\begin{center}
$K^{A,D}_s(y) = \min\{\eta r, K^A_s(y)\} + O(\log r)$.
\end{center}
If $K^{A,D}_s(y) = K^A_s(y) - O(\log r)$, then 
\begin{align}
K^{A,D}_r(z) &\geq K^A_s(y) + r - s -\frac{\ve r}{2}- O(\log r)\tag*{}\\
&> ds - \ve r + r - s - \ve r \tag*{[(C3)]}\\
&= dr - \left((d-1) (r-s)\right) - 2\ve r\tag*{}\\
&\geq dr - \frac{(d-1)r}{2} - 3\ve r\tag*{}\\
&= \frac{(d+1)r}{2} - 4\ve r \tag*{}\\
&\geq (\eta + \ve) r.
\end{align}
Alternatively, if $K^{A,D}_s(y) = \eta r - O(\log r)$,
\begin{align*}
K^{A,D}_r(z) &\geq  \eta r + r - s -\frac{\ve r}{2} - O(\log r)\\
&\geq \eta r + r- s - \ve r
\end{align*}
Therefore, if $s \geq \frac{r}{2} - \log r$,
\begin{equation}\label{eq:mainThmEffDim5}
K^{A,D}_r(z) \geq \eta r + \min\{\ve r, r - s -\ve r\}.
\end{equation}

\bigskip

\noindent \textit{Case 2: $s < \frac{r}{2} - \log r$.}  Let $e^\prime = \frac{y-z}{\|y-z\|}$. In this case, we will need to apply Theorem \ref{thm:projectionMainTheorem}. We begin by showing that the conditions of Theorem \ref{thm:projectionMainTheorem} are satisfied for $x, e^\prime, \ve, s, r-s$ and $C^\prime$, where $C^\prime$ depends only on $x, y$ and $e$.

Since $s < \frac{r}{2} - \log r$, and $r$ is assumed to be sufficiently large, $r - s$ is sufficiently large. Conditions (P1) and (P2) are clearly satisfied.  To see (P3), by Observation \ref{obs:geometricObs}, 
\begin{align*}
\|e^\bot - e^\prime\| &\leq \frac{\|y-z\|}{\|y-x\|}\\
&= \frac{2^{-s}}{\|x-y\|}.
\end{align*}
Therefore, for every $s^\prime \leq s$, by (C2),
\begin{align}
K^{A,x}_{s^\prime}(e^\prime) &= K^{A,x}_{s^\prime}(e^\bot) + O(\log s^\prime)\tag*{[Lemma \ref{lem:MD:3.9}]}\\
&= K^{A,x}_{s^\prime}(e) + O(\log s^\prime)\tag*{[$e$ computable from $e^\bot$]}\\
&\geq s^\prime - C^\prime\log s^\prime\tag*{[(C2)]}.
\end{align}
Thus the conditions for Theorem \ref{thm:projectionMainTheorem}, relative to $A$, are satisfied for $x, e^\prime, \ve, s, r-s$ and $C^\prime$. 

We now turn to giving a lower bound on $K^{A,D}_r(z)$. By Lemma \ref{lem:lowerBoundOtherPointDistance},
\begin{align}
K^{A,D}_r(z) &\geq K^{A,D}_s(y) + K^{A,D}_{r-s, r}(x\mid y) \tag*{[Lemma \ref{lem:lowerBoundOtherPointDistance}]}\\
&\;\;\;\;\;\;\;\;\;\;\;\;\;\;- K^{A,D}_{r-s}(x\mid p_{e^\prime}, e^\prime) - O(\log r)\tag*{}\\
&\geq K^{A,D}_s(y) + K^A_{r-s, r}(x\mid y)  \tag*{[Lemma \ref{lem:oracles}]}\\
&\;\;\;\;\;\;\;\;\;\;\;\;\;\; - K^{A,D}_{r-s}(x\mid p_{e^\prime}, e^\prime) - O(\log r)\tag*{}\\
&\geq K^{A,D}_s(y) + K^{A}_{r-s}(x)  - K^{A,D}_{r-s}(x\mid p_{e^\prime}, e^\prime) - O(\log r)\tag*{[Lemma \ref{lem:symmetry}(ii)]}\\
&\geq K^{A,D}_s(y) + K^{A}_{r-s}(x)  - K^{A}_{r-s}(x\mid p_{e^\prime}, e^\prime) - O(\log r)\label{eq:mainThmEffDim2}.
\end{align}

Applying Theorem \ref{thm:projectionMainTheorem}, relative to $A$, yields
\begin{center}
$K^{A}_{r-s}(x\mid p_{e^\prime} x, e^\prime) \leq K^A_{r-s}(x) - \frac{r}{2} + \ve (r-s)$.
\end{center}
By combining this with (\ref{eq:mainThmEffDim2}) we see that
\begin{equation}\label{eq:mainThmEffDim3}
K^{A,D}_r(z)\geq K^{A,D}_s(y) + \frac{r}{2} - \ve r - O(\log r).
\end{equation}
Since $s \geq t$, by (\ref{eq:boundComplexityAtT}), $K^{A,D}_s(y) \geq \frac{dr}{2} - O(\log r)$. Thus,
\begin{align}
K^{A,D}_r(z) &\geq K^{A,D}_s(y) + \frac{r}{2} - \ve r- O(\log r)\tag*{}\\
&\geq  \frac{(d+1)r}{2} - 2\ve r \tag*{}\\
&\geq (\eta + \ve) r\label{eq:mainThmEffDim6}.
\end{align}

\bigskip

Therefore, by (\ref{eq:mainThmEffDim5}) and (\ref{eq:mainThmEffDim6}) the conditions of Lemma \ref{lem:pointDistance} are satisfied. Applying this, relative to $(A,D)$, yields
\begin{equation*}
K_{r,t}^{A,D,x}(y \mid y) \leq K^{A,D,x}_{r,t}(\|x-y\|\mid y) +3\ve r + K(\ve,\eta)+O(\log r).
\end{equation*}
Rearranging and using the fact that additional information does not increase complexity shows that
\begin{equation}\label{eq:mainThmEffDim4}
K^{A,x}_{r,t}(\|x-y\|\mid y) \geq K^{A,D,x}_{r,t}(y\mid y) -4\ve r.
\end{equation}
A standard symmetry of information argument bound the right hand side. Formally,
\begin{align}
K^{A,D,x}_{r,t}(y\mid y) &=  K^{A,D,x}_{r}(y) - K^{A,D,x}_{t}(y) - O(\log r)\tag*{[Lemma \ref{lem:unichain}]}\\
&= K^{A,D}_r(y) - K^{A,D}_t(y) - O(\log r) \tag*{[Lemma \ref{lem:oracles}(iii), (C4]}\\
&= \frac{r}{2} - 4\ve r - O(\log r)\tag*{}.
\end{align}
Combining this with (\ref{eq:mainThmEffDim4}) and (\ref{eq:lowerBoundDistanceAtT}) we have
\begin{align*}
K^{A,x}_r(\|x-y\|) &= K^{A,x}_{t}(\|x-y\|) + K^{A,x}_{r,t}(\|x-y\|\mid \|x-y\|) \\
&\geq K^{A,x}_{t}(\|x-y\|) + K^{A,x}_{r,t}(\|x-y\|\mid y) \\
&\geq \frac{dr}{4} + K^{A,D,x}_{r,t}(y\mid y) -5\ve r\\
&\geq \frac{dr}{4} + \frac{r}{2} - 9\ve r - O(\log r)\\
&= \frac{dr}{4} + \frac{r}{2} - 9\ve r - O(\log r).
\end{align*}
Since this bound holds for all sufficiently large $r$, we conclude that
\begin{center}
$\dim^{A,x}(\|x-y\|) \geq \frac{d}{4}  + \frac{1}{2} - 9\ve$.
\end{center}
Since $\ve$ was chosen arbitrarily, the conclusion follows.
\end{proof}

\section{Dimension bounds for distance sets}
In this section, we prove the main theorem of the paper. The strategy is to use the point-to-set principle to reduce this to the analogous bound for effective dimension, Theorem \ref{thm:mainThmEffDim}. We need to show that, for all $x$ outside a set of dimension at most $1$, there is a $y\in E$ so that $x$ and $y$ satisfy the following conditions, relative to an appropriate oracle $A$.
\begin{itemize}
\item[\textup{(C1)}] $\dim^{A}(x) > 1$
\item[\textup{(C2)}] $K^{x,A}_r(e) = r - O(\log r)$ for all $r$.
\item[\textup{(C3)}] $\dim^A(y) \geq s$.
\item[\textup{(C4)}] $K^{x,A}_r(y) \geq K^{A}_r(y) - O(\log r)$ for all sufficiently large $r$. 
\item[\textup{(C5)}] $K^{A}_r(e\mid y) = r - o(r)$ for all $r$.
\end{itemize}
To prove the existence of points $x$ and $y$ with these conditions, we combine an important result of Orponen \cite{Orponen19DimSmooth} on radial projections with recent work in algorithmic information theory.

\begin{lem}\label{lem:reductionToEff}
Let $E\subseteq\R^2$ be a compact set such that $0<\mathcal{H}^s(E)<\infty$ for some $s > 1$. Let $A\subseteq\N$ be any oracle relative to which $E$ is effectively compact. Then there is a set $B\subseteq\R^2$ of Hausdorff dimension at most $1$ such that, for every  $x\in\R^2 - B$, there is a $y\in E$ such that $x,y$ satisfy (C1)-(C5).
\end{lem}

Assuming Lemma \ref{lem:reductionToEff}, we can prove the main theorem of this paper.
\begin{T1}
Let $E\subseteq\R^2$ be analytic such that $\dim_H(E) = s > 1$. Then for all $x\in\R^2$ outside a set $B$ with $\dim_H(B)\leq 1$,
\begin{center}
$\dim_H(\Delta_x E) \geq \frac{s}{4} + \frac{1}{2}$.
\end{center}
\end{T1}
\begin{proof}
Let $E\subseteq\R^2$ be an analytic set such that $\dim_H(E) > 1$. Then for every $1<s<\dim_H(E)$, there is a compact subset $F\subseteq E$  such that $0 < \mathcal{H}^s(F) < \infty$. For every $n\in\N$, let $s_n = s - 1/n$, and let $F_n$ such a compact subset of $E$ with respect to $s_n$.  Let $A_n \subseteq \N$ be an oracle relative to which $F_n$ is effectively compact. Finally, Let $A$ be the oracle which is the join of all $A_n$. Note that, for every $n\in\N$, $F_n$ is effectively compact relative to $A$.

For each $n$, let $B_n \subseteq \R^2$ be the oracle and the set guaranteed by Lemma \ref{lem:reductionToEff}, with respect to $F_n$. Let $B$ be the union of all $B_n$, noting that $\dim_H(B) \leq 1$. Then, for every $x \in \R^2 - B$, and every $n$, there is a point $y\in F_n$ such that $x,y$ satisfy the conditions (C1)-(C5). Therefore, for all $x \in\R^2 - B$, and every $n\in\N$, by Theorem \ref{thm:mainThmEffDim}
\begin{center}
$\dim^{A,x}(\|x-y\|) \geq \frac{s}{4} + \frac{1}{2} - \frac{1}{4n}$.
\end{center}

Finally, let $n\in\N$. Since $F_n$ is effectively compact relative to $A$, $\Delta_x F_n$ is effectively compact relative to $(A,x)$, for every $x\in \R^2$ (Proposition \ref{prop:effcompact}). Therefore, by Theorem \ref{thm:strongPointToSetDim}, 
\begin{align*}
\dim_H(\Delta_x E)&\geq \dim_H(\Delta_x F_n) \\
&= \sup\limits_{y\in E} \dim^{A,x}(\|x-y\|)\\
&\geq \frac{s}{4} + \frac{1}{2} - \frac{1}{4n}.
\end{align*}
Since $n$ was arbitrary, the conclusion follows.
\end{proof}

\subsection{Proof of Lemma \ref{lem:reductionToEff}}

We first show that condition (C5) follows from (C1)-(C4). 
\begin{lem}\label{lem:c5fromc14}
If $x,y\in\R^2$ satisfy (C1)-(C4) with respect to $A\subseteq\N$, then they also satisfy (C5).
\end{lem}
\begin{proof}
Assume that $x,y\in \R^2$ satisfy (C1)-(C4). Since 
\begin{center}
$K^{A,x}_r(e) = K^{A,x}_r(e^\bot) + O(\log r)$
\end{center}
for every $r\in\N$, by the main result of \cite{LutStu18}, 
\begin{center}
$\dim^{A,e^\bot}(p_{e^\bot} x) = 1$.
\end{center}
By the definition of of effective dimension, we have
\begin{equation}\label{eq:c1c2ProjLowerBound}
K^A_r(p_{e^\bot} x \mid e^\bot) \geq r - o(r).
\end{equation}

Let $r\in\N$. Then, by symmetry of information, and conditions (C2) and (C4),
\begin{align}
K^A_r(x) + K^A_r(y) &= K^A_r(x) + K^A_r(y\,\mid\, x) +O(\log r)\tag*{}\\
&= K^A_r(x, y)+O(\log r)\tag*{}\\
&= K^A_r(x,y,e)+O(\log r)\tag*{}\\
&= K^A_r(x) + K^A_r(e \, | \, x)  + K^A_r(y \, |\, e, x)+O(\log r)\tag*{}\\
&= K^A_r(x) + r + K^A_r(y \, | \, e, x)+O(\log r)\tag*{}.
\end{align}
Rearranging, we see that 
\begin{equation}\label{eq:c5fromc14eq1}
K^A_r(y) = K^A_r(y \, |\, e, x) + r+O(\log r).
\end{equation}
Note that $p_{e^\bot}x = p_{e^\bot} y$ and $K^A_r(y \, | \, e, x) = K^A_r(y \, | \, e^\bot, x) + O(\log r)$. Thus, by (\ref{eq:c1c2ProjLowerBound}),
\begin{align}
K^A_r(y \, |\, e^\bot, x) &\leq K^A_r(y \, | \, e^\bot, p_{e^\bot} x)+O(\log r)\tag*{}\\
&= K^A_r(y \, | \, e^\bot) - K^A_r(p_{e^\bot} x \, | \, e^\bot)+O(\log r)\tag*{}\\
&= K^A_r(y \, | \, e^\bot) - r + o(r)\label{eq:c5fromc14eq2}, 
\end{align}
Combining (\ref{eq:c5fromc14eq2}) and (\ref{eq:c5fromc14eq1}), we conclude that
\begin{align}
K^A_r(y) &= K^A_r(y \, |\, e, x) + r+O(\log r)\tag*{}\\
&= K^A_r(y \, | \, e^\bot, x) + r+O(\log r)\tag*{}\\
&\leq K^A_r(y \, | \, e^\bot) + o(r)\tag*{}\\
&= K^A_r(y \, | \, e) + o(r)\tag*{}\\
&\leq K^A_r(y) + o(r)\tag*{}.
\end{align}

This shows that $e^\bot$, and therefore $e$, contains (essentially) no information about $y$. So, by the symmetry of information, $y$ does not contain information about $e$. Formally,
\begin{align}
K^A_r(y) + K^A_r(e  \mid y) &= K^A_r(y, e)+O(\log r)\tag*{}\\
&= K^A_r(e) + K^A_r(y\mid e) +O(\log r)\tag*{}\\
&= K^A_r(e) + K^A_r(y)+ o(r)\tag*{}\\
\end{align}
Therefore, $K^A_r(e) = K^A_r(e \mid  y) + o(r)$, and the proof is complete. 
\end{proof}

We will also need a theorem of Orponen \cite{Orponen19DimSmooth} on radial projections. For any $x\in \R^2$, define the \textit{radial projection} $\pi_x : \R^2 - \{x\} \to \mathcal{S}^1$ by 
\begin{center}
$\pi_x(y) = \frac{y-x}{\|x-y\|}$.
\end{center}
For any $x$ and measure $\mu$, the pushforward of $\mu$ under $\pi_x$ is the measure given by
\begin{center}
$\pi_{x\#} \mu (S) = \mu\left( \pi_x^{-1}(S)\right)$.
\end{center}
If the pushforward of $\mu$ under $\pi_x$ is absolutely continuous with respect to $\mathcal{H}^1\vert_{\mathcal{S}^1}$, then for all $S$ such that $\mathcal{H}^1\vert_{\mathcal{S}^1}(S) = 0$, we must have $\mu(\pi_x^{-1}(S)) = 0$.

\begin{thm}[\cite{Orponen19DimSmooth}]\label{thm:orponenradial}
Let $E \subseteq \R^2$ be a Borel set with $s = \dim_H(E) > 1$ such that there is a measure $\mu\in \mathcal{M}(E)$ satisfying $I_d(\mu) < \infty$ for all $1 < d < s$. Then there is a Borel set $B \subseteq \R^2$ with $\dim_H(B) \leq 2 - \dim_H(E)$ such that, for every $x \in \R^2 - B$, $\mathcal{H}^1(\pi_x(E)) > 0$. Moreover,  the pushforward of $\mu$ under $\pi_x$ is absolutely continuous with respect to $\mathcal{H}^1 \vert_{\mathcal{S}^1}$ for $x \notin B$.
\end{thm}

In order to prove Lemma \ref{lem:reductionToEff}, we need to show that, for ``most" points $x\in \R^2$, there is some $y\in E$ so that $x, y$ satisfy (C1)-(C5). By Lemma \ref{lem:c5fromc14}, this reduces to showing that $x, y$ satisfy (C1) - (C4). We  combine Theorem \ref{thm:orponenradial} with results on independence in algorithmic randomness to prove this.
\begin{lem}\label{lem:existenceGoodPoints}
Let $E \subseteq \R^2$ be compact such that $0 < \mathcal{H}^s(E) < \infty$ for some $s > 1$. Let $\mu = \mathcal{H}^s\vert_{E}$. Let $A\subseteq\N$ such that $A$ is computably compact relative to $E$, and $\mu$ is computable relative to $A$. Then for all $x\in \R^2 - spt(\mu)$ outside a set $B$ of Hausdorff dimension at most $1$, there is a $y \in E$ such that $x,y$ satisfy (C1)-(C5) with respect to $A$. 
\end{lem}
\begin{proof}
Let $B$ be the set guaranteed by Theorem \ref{thm:orponenradial}, with respect to $E$ and $\mu = \mathcal{H}^s\vert_E$. Let $x \in \R^2 - spt(\mu) - B$ be any point such that $\dim^A(x) > 1$. Let 
\begin{center}
$N = \{e \in \mathcal{S}^1 \, | \, (\exists^\infty r) \,  K^{x,A}_r(e) < r - 4\log r\}$.
\end{center}
Note that $N$ is Borel and $N$ has $\mathcal{H}^1\vert_{\mathcal{S}^1}$ measure zero (Lemma \ref{lem:independenceResults}). Thus, by Theorem \ref{thm:orponenradial}, $\mu(\pi_x^{-1}(N)) = 0$. Let $F$ be the set of points $y \in E$  such that
\begin{itemize}
\item $\dim^A(y) \geq s$, and 
\item $K^{x,A}_r(y) > K^A_r(y) - 8\log r$ for all sufficiently large $r$.
\end{itemize}
Note that $F$ is Borel, and therefore measurable. It is implicit in \cite{Stull22a} that $\mathcal{H}^s(F)  = \mathcal{H}^s(E) > 0$. For completeness, we give a proof of this fact in the appendix (Lemma \ref{lem:independenceResults}). Therefore, there is a point $y\in F - \pi_x^{-1}(N)$. It is immediate by our choice of $x$ and $y$ that they satisfy (C1)-(C4) relative to $A$. Hence, by Lemma \ref{lem:c5fromc14}, $x$ and $y$ also satisfy (C5), and the proof is complete.
\end{proof}

We are now able to prove Lemma \ref{lem:reductionToEff}, thereby completing the proof of the main theorem.
\begin{L1}
Let $E\subseteq\R^2$ be a compact set such that $0<\mathcal{H}^s(E)<\infty$ for some $s > 1$. Let $A\subseteq\N$ be any oracle relative to which $E$ is effectively compact. Then there is a set $B\subseteq\R^2$ of Hausdorff dimension at most $1$ such that, for every  $x\in\R^2 - B$, there is a $y\in E$ such that $x,y$ satisfy (C1)-(C5).
\end{L1}
\begin{proof}
Let $E\subseteq \R^2$ be compact with $0< \mathcal{H}^s(E) <\infty$, for some $s > 1$. Then there are disjoint, compact subsets $E_1, E_2\subseteq E$  such that $0<\mathcal{H}^s(E_1), \mathcal{H}^s(E_2) < \infty$. Let $\mu_1 = \mathcal{H}^s\vert_{E_1}$ and $\mu_2 = \mathcal{H}^s\vert_{E_2}$. Let $A_1$ and $A_2$ be oracles relative to which $E_1$ and $E_2$ are effectively compact. Let $A$ be an oracle relative to which $E$ is effectively compact. Let $A$ be the join of $A_1, A_2, A_3, \hat{\mu}_1$ and $\hat{\mu}_2$, where $\hat{\mu}_i$ is an oracle encoding $\mu_i(Q)$ for every ball $Q$ of rational radius with a rational center. 

For $i\in\{1,2\}$, let $B_i$ be the set of exceptional points so that, for all $x\notin B_i \cup \text{spt}(\mu_i)$, there is a $y\in E_i$ so that $x$ and $y$ satisfy conditions (C1)-(C5) with respect to $A$ (guaranteed by Lemma \ref{lem:existenceGoodPoints}). 

Let 
\begin{center}
$B = \{x\in \R^2\mid \dim^A(x) \leq 1\} \cup_i B_i$.
\end{center}
Note that $\dim_H(B) \leq 1$. Let $x\in\R^2 - B$. Then by definition, $x$ satisfies (C1). Moreover, there is a point $y\in E_1\cup E_2$ so that $x,y$ satisfies conditions (C1)-(C5) with respect to $A$, and the claim follows.
\end{proof}

\bibliography{pdsc}

\appendix
\renewcommand{\thesection}{\Alph{section}}
\section{Appendix}

\begin{obs}\label{obs:geometricObs}
Let $x,y, z\in\R^2$ such that $\|x - y\| = \|x-z\|$. Let $e_1 = \frac{y - x}{\|x - y\|}$ and $e_2 = \frac{y-z}{\|y-z\|}$. Let $w\in\R^2$ be the midpoint of the line segment from $y$ to $z$, and $e_3 = \frac{w-x}{\|x-w\|}$. Then 
\begin{enumerate}
\item $|p_{e_1} y - p_{e_1} z| = \frac{\|y - z\|^2}{2\|x-y\|}$, and
\item $\|e_1 - e_3\| < \frac{\|y-z\|}{\|x-y\|}$.
\end{enumerate}
\end{obs}
\begin{proof}
Consider the triangle with endpoints $x, y$ and $z$. By translating by $x$ and rotating we may assume that $e_1 = (0, 1)$. Then $|p_{e_1} y - p_{e_1} z| = \|x-y\| (1 - \cos \theta)$, where $\theta$ is the angle between vectors $y$ and $z$. Then we have
\begin{align*}
\cos \theta &= \frac{\|x-y\|^2 + \|x-z\|^2 - \|y-z\|^2}{2\|x-y\|\, \|x-z\|}\\
&= \frac{2\|x-y\|^2 - \|y-z\|^2}{2\|x-y\|^2}\\
&= 1 - \frac{\|y-z\|^2}{2\|x-y\|^2},
\end{align*}
and item (1) follows. 

To see (3), note that
\begin{align*}
\|e_1 - e_3\| &< \|e_1 - \frac{z-x}{\|x-z\|}\|\\
&= \frac{\|y - z\|}{\|x-y\|}.
\end{align*}
\end{proof}

\begin{obs}\label{obs:existProj}
Let $x\in \R^2$, $d \in \R_{> 0}$, $p \in \Q^2$, and $r \in \N$ such that $\vert d  - \|x - p\| \vert < 2^{-r}$. Then there is a $y \in \R^2$ such that $\| p - y\| < 2^{-r}$ and $d = \|x-y\|$.
\end{obs}
\begin{proof}
Suppose $x\in \R^2$, $d \in \R_{> 0}$, $p \in \Q^2$, and $r \in \N$ satisfy the hypothesis. Define
\[y = p + \left(d - \|x-p\|\right)\frac{p-x}{\|x-p\|}\,.\]

Then we deduce that
\begin{align*}
    \| x - y\| &= \| x- p -\left(d - \|x-p\|\right)\frac{p-x}{\|x-p\|}\\
   &= \| x-p + \frac{x-p}{\|x-p\|}\left(d - \|x-p\|\right)\|\\
    &=\| (x-p) \frac{d}{\|x-p\|}\|\\
    &= d.
\end{align*}
Moreover, 
\begin{align*}
    \| y-p\| &= \|\left(d - \|x-p\|\right)\frac{p-x}{\|x-p\|}\|\\
    &= \vert d - \|x-p\| \vert\\
    &< 2^{-r}.
\end{align*}

\end{proof}

The following lemma is needed in the proof of Lemma \ref{lem:existenceGoodPoints}. Item (1) is a standard argument in algorithmic randomness. Item (2) is implicit in \cite{Stull22a}, and is similar to methods used in \cite{LutLut20AlgOpt}.
\begin{lem}\label{lem:independenceResults}
Let $E \subseteq \R^2$ be compact and $0 < \mathcal{H}^s(E) < \infty$ for some $s > 1$. Let $\mu = \mathcal{H}^s\vert_{E}$. Let $A\subseteq\N$ such that $A$ is computably compact relative to $E$, and $\mu$ is computable relative to $A$. Then the following hold.
\begin{enumerate}
\item The set
\begin{center}
$N_1 = \{e \in \mathcal{S}^1 \, | \, (\exists^\infty r) \,  K^{x,A}_r(e) < r - 4\log r\}$.
\end{center}
has $\mathcal{H}^1\vert_{\mathcal{S}^1}$ measure $0$.
\item The set
\begin{center}
$N_2 = \{y\in E \mid (\exists^\infty r)\, K^{A,x}_r(y) < K^A_r(y) - 8\log r\}$
\end{center}
has $\mu$-measure $0$.
\end{enumerate}
\end{lem}
\begin{proof}
We begin by proving (1). For every $r\in\N$, let 
\begin{center}
$U_r = \{B_{2^{-r}}(q) \mid q\in\Q_+ \text{ and } K^{A,x}(q) < r - 4\log r\}$.
\end{center}
Then $\vert U_r\vert \leq 2^{r-4\log r}$, since there are at most $2^{r-4\log r}$ programs of length less than $r-4\log r$. By definition, for every $R\in\N$, $\bigcup_{r> R} U_r$ is an open cover of $N_1$. Moreover,
\begin{align*}
\mathcal{H}^1\vert_{\mathcal{S}^1}(N_1) &\leq \sum\limits_{r > R} \sum\limits_{B\in U_r} 2^{-r}\\
&\leq \sum\limits_{r > R} 2^{r-4\log r}2^{-r}\\
&= \sum\limits_{r > R} \frac{1}{r^4},
\end{align*}
and the conclusion follows.

The proof of (2) is similar to (1), but requires Levin's coding theorem \cite{Levin73, Levin74}, a deep result of algorithmic information theory. Let $\mathcal{Q}_r$ be the set of all half-open dyadic cubes at precision $r$. Each dyadic cube $Q$ can be uniquely identified with dyadic rational. Let $p_\mu:\mathbb{D}^2 \rightarrow \R_+$ be the following discrete semimeasure defined on the dyadic rationals 
\begin{center}
$p_\mu (d) = \frac{\mu(Q_d)}{\mu(E)}$,
\end{center}
where $Q$ is the unique dyadic half-open box centered at $d$. Given oracle access to $A$,  $p_\mu$ is lower semicomputable. Case and Lutz \cite{CasLut15}, generalizing Levin's coding theorem, showed that, for every precision $r$ and every dyadic $d$ of precision $r$,
\begin{equation}\label{eq:levincoding}
2^{-K^A(d) + K^A(r) +c} > p_\mu(d),
\end{equation}
for some fixed constant $c$.

With this machinery in place, the rest of the proof is similar to the argument for (1). For every $r\in\N$, let 
\begin{center}
$W_r = \{Q\in \mathcal{Q}_r \mid K^{A,x}(d_Q) < K^A(d_Q) - 4\log r\}$.
\end{center}
Note that, by \cite{CasLut18}, $K^{A,x}_r(y) = K^{A,x}_r(d_Q) + 4\log r$ for every precision $r$, where $Q\in\mathcal{Q}_r$ is the dyadic cube at precision $r$ containing $y$ . Therefore, the set $\bigcup_{r>R} W_r$ is a cover of $N_2$. Applying Kraft's inequality and (\ref{eq:levincoding}) we deduce that
\begin{align}
1 &\geq \sum\limits_{r > R} \sum\limits_{Q\in W_r} 2^{-K^{A,x}(d_Q)}\tag*{[Kraft inequality]}\\
&\geq \sum\limits_{r > R} \sum\limits_{Q\in W_r} 2^{-K^{A}(d_Q) + 4\log r}\tag*{}\\
&> \sum\limits_{r > R} \sum\limits_{Q\in W_r} 2^{4\log r - K^A(r) - c}p_\mu(d_Q)\tag*{}\\
&\geq \sum\limits_{r > R} \sum\limits_{Q\in W_r} 2^{2\log r - c}p_\mu(d_Q)\tag*{}\\
&= \sum\limits_{r > R} \sum\limits_{Q\in W_r} 2^{2\log r - c} \frac{\mu(Q_d)}{\mu(E)}\tag*{}\\
&> \frac{R}{2^c\mu(E)}\sum\limits_{r > R} \sum\limits_{Q\in W_r} \mu(Q_d)\tag*{}\\
&\geq \frac{R}{2^c\mu(E)} \mu(N_2)\tag*{},
\end{align}
and the conclusion follows.
\end{proof}
\end{document}